\documentclass[a4paper,UKenglish]{lipics2019changed}

\usepackage{microtype}
\usepackage{mathtools,stmaryrd}
\usepackage{xspace}

\usepackage{todonotes}

\usepackage{tikz}

\usetikzlibrary{matrix}

\usepackage{amssymb, latexsym} 

\usetikzlibrary{arrows}
\usetikzlibrary{shapes}

\newcommand{\nc}{\newcommand}

\nc{\theo}{\begin{theorem}}
\nc{\etheo}{\end{theorem}}

\nc{\lem}{\begin{lemma}}
\nc{\elem}{\end{lemma}}
\nc{\elemma}{\end{lemma}}

\nc{\corol}{\begin{corollary}}
\nc{\ecorol}{\end{corollary}}

\theoremstyle{definition}
\newtheorem{defn}[theorem]{Definition}
\nc{\bdefn}{\begin{defn}}
\nc{\edefn}{\end{defn}}

\nc{\define}[1]{{\bf\boldmath #1}}

\theoremstyle{plain}
\nc{\prop}{\begin{proposition}}
\nc{\eprop}{\end{proposition}}

\newtheorem{Conjecture}[theorem]{Conjecture}
\newcommand{\conj}{\begin{Conjecture}}
\newcommand{\econj}{\end{Conjecture}}

\newtheorem{Example}[theorem]{\small Example}
\newcommand{\ex}{\begin{Example}\small \rm}
\newcommand{\eex}{\end{Example}}

\newtheorem{Examples}{\small Examples}
\nc{\exx}{\begin{Examples}\small\rm\begin{enum}}

\newcounter{enum}
\newenvironment{enum}{\begin{list}{(\arabic{enum})}%
{\setlength{\labelwidth}{5mm}\setlength{\leftmargin}{10mm}%
\setlength{\itemindent}{0pt}\usecounter{enum}}}{\end{list}}

\newcounter{menum}

\newcounter{aenum}

\newcommand{\look}{\begin{proof}}
\newcommand{\lueg}{\look}
\newcommand{\hx}{\end{proof}}

\newcommand{\ra}{\rightarrow}
\renewcommand{\iff}{\leftrightarrow}
\newcommand{\E}{\exists}
\newcommand{\A}{\forall}
\renewcommand{\phi}{\varphi}
\renewcommand{\theta}{\vartheta}
\renewcommand{\emptyset}{\varnothing}

\renewcommand{\AA}{{\mathfrak A}}

\renewcommand{\epsilon}{\varepsilon}

\newcommand{\FO}{{\rm FO}}

\newcommand{\LFP}{{\rm LFP\ }}

\newcommand{\posLFP}{{\rm posLFP}}

\newcommand{\N}{{\mathbb N}}
\newcommand{\Ninf}{{\mathbb N}^{\infty}}
\newcommand{\Sinf}{{\mathbb S}^{\infty}}
\newcommand{\R}{{\mathbb R}}

\newcommand{\B}{{\mathbb B}}
\renewcommand{\S}{{\mathbb S}}
\newcommand{\V}{{\mathbb V}}
\newcommand{\W}{{\mathbb W}}

\nc{\IF}{{\bf if }}
\nc{\THEN}{{\bf then } }
\nc{\ELSE}{{\bf else } }
\nc{\DO}{{\bf do }}
\nc{\OD}{{\bf od } }
\nc{\OUTPUT}{{\bf output }}
\nc{\END}{{\bf end }}
\nc{\WHILE}{{\bf while} }
\nc{\INPUT}{{\bf Input} }
\nc{\GUESS}{{\bf guess } }
\nc{\CHOOSE}{{\bf universally choose } }
\nc{\REJECT}{{\bf reject } }
\nc{\ACCEPT}{{\bf accept } }
\nc{\FOR}{{\bf for } }
\nc{\FORALL}{{\bf for all } }
\nc{\TO}{{\bf to } }
\nc{\RETURN}{{\bf return } }

\newcommand{\Gg}{{\cal G}}

\newcommand{\Kk}{{\cal K}}

\newcommand{\Ss}{{\cal S}}
\newcommand{\Tt}{{\cal T}}

\renewcommand{\bar}{\overline}

\newcommand{\cutout}[1]{}

\DeclareMathOperator{\lfp}{\mathbf{lfp}}
\DeclareMathOperator{\gfp}{\mathbf{gfp}}

\nc{\ext}[1]{[\![ #1 ]\!]}  \nc{\extI}[1]{[\![ #1 ]\!]^\II}
\nc{\extK}[1]{[\![ #1 ]\!]^\Kk} \nc{\extT}[1]{[\![ #1 ]\!]^\Tt}

\renewcommand{\st}{\,.\,} 

\newcommand{\Strat}{\mathrm{Strat}}
\DeclareMathOperator{\Plays}{\mathrm{Plays}}

\DeclareMathOperator{\Atoms}{\mathrm{Atoms}}
\DeclareMathOperator{\NegAtoms}{\mathrm{NegAtoms}}

\DeclareMathOperator{\Lit}{\mathrm{Lit}}
\DeclareMathOperator{\nnf}{\mathrm{nnf}}
\DeclareMathOperator{\op}{\mathrm{\,op\,}}

\nc{\Exptime}{\text{\sc Exptime}}
\nc{\Nexptime}{\text{\sc Nexptime}}
\nc{\Ntime}{\text{\sc Ntime}}
\nc{\Ptime}{\text{\sc Ptime}}

\newcommand{\Bool}{\mathbb{B}}
\newcommand{\Access}{\mathbb{A}}
\newcommand{\Nat}{\mathbb{N}}

\newcommand{\Trop}{\mathbb{T}}
\newcommand{\Vit}{\mathbb{V}}

\newcommand{\PosBool}{\mathsf{PosBool}}

\newcommand{\Pub}{\mathsf{P}}
\newcommand{\Cnf}{\mathsf{C}}
\newcommand{\Sec}{\mathsf{S}}
\newcommand{\Tsec}{\mathsf{T}}

\newcommand{\nn}[1]{\bar{#1}}
\newcommand{\nnp}{\nn{p}}

\newcommand{\nnX}{\nn{X}}
\newcommand{\nnx}{\nn{x}}

\newcommand{\topic}[1]{\medskip\noindent{\bf #1.\ }}

\title{Provenance Analysis for Logic and Games}
\author{Erich Grädel}{RWTH Aachen University, Germany}{graedel@logic.rwth-aachen.de}{}{}
\author{Val Tannen}{Univ.~of~Pennsylvania, U.S.A.}{val@cis.upenn.edu}{}{}

\authorrunning{E. Grädel and V. Tannen}

\Copyright{Erich Grädel and Val Tannen}



\nolinenumbers 
\hideLIPIcs  

\begin{document}

\maketitle

\begin{abstract}
A model checking computation checks whether a given logical sentence is true in a given finite structure.
Provenance analysis abstracts from such a computation mathematical information on how the result depends on the
atomic data that describe the structure. 
In database theory, provenance analysis by interpretations in commutative semirings has been rather succesful
for positive query languages (such a unions of conjunctive queries, positive relational algebra, or datalog).
However, it did not really offer an adequate treatment of negation or missing information.
Here we propose a new approach for the provenance analysis of logics with negation, such as first-order logic and 
fixed-point logics. It is closely related to a provenance analysis of the associated model-checking games, 
and based on new semirings of  dual-indeterminate polynomials or dual-indeterminate formal power series.
These are obtained by taking quotients of traditional provenance semirings by congruences 
that are generated by products of positive and negative provenance tokens.
Beyond the use for model-checking problems in logics, provenance analysis of games is of independent interest.
Provenance values in games provide detailed information about the number and properties of the strategies of the players, 
far beyond  the question whether or not a player has a winning strategy from a given position.

\end{abstract}

\section{Introduction}

Provenance analysis aims at understanding how the result of a computational process
with a complex input, consisting of multiple items, depends on the various parts of
this input. In database theory, provenance analysis based on interpretations in commutative semirings 
has been developed for positive database query languages,  to understand which combinations of 
the atomic facts in a database can be used for deriving the result of a given query. 
In this approach, atomic facts are interpreted not just by true or false, but by values in an appropriate
semiring, where 0 is the value of false statements, whereas any element $a\neq 0$ of the  semiring  stands
for some shade of truth. These values are then propagated from the atomic facts to arbitrary queries in the language, 
which permits to answer questions such as the minimal cost of a query evaluation, the confidence one can
have that the result is true, the number of different ways in which the result can be computed, 
or the clearance level that is required for obtaining the output,
under the assumption that some facts are labelled as confidential, secret, top secret, etc.
We refer to \cite{GreenTan17} for a recent account and many references on the semiring framework for database provenance.

Scenarios to which the  semiring provenance approach has been successfully applied include  
unions of conjunctive queries, positive relational algebra, nested relations, Datalog, XQuery, SQL-aggregates 
and several others,
and it has been implemented in software systems such as Orchestra
and Propolis. For details, see e.g. \cite{AmsterdamerDeuTan11a,DeutchMilRoyTan14,FosterGreTan08,Green11,GreenKarTan07,
Tannen13}.   A main limitation of this approach is that is has been
largely confined to \emph{positive} query
languages. Attempts to add operations that capture \emph{difference of relations} have led to
interesting and algebraically challenging, but divergent
approaches \cite{AmsterdamerDeuTan11,GeertsPog10,GeertsUngKarFunChr16,GreenIveTan09}. 
In particular there has been no systematic approach in database theory for tracking \emph{negative information}, 
and no convincing provenance analysis for languages with full negation.

Here, we would like to develop a new approach for a semiring provenance analysis
for model checking problems of logics with negation, in particular first-order logic and fixed-point logic.
This approach is based on several ideas:
\begin{itemize}
\item Provenance analysis of logics is intimately connected to provenance analysis of games. In the same way as
formula evaluation or model checking can be formulated in game theoretic terms, also the propagation of provenance values 
from atomic facts to arbitrary formulae can be viewed
as a process on the associated games. Also the typical \emph{results} of a provenance analysis of database queries
or logical formulae, concerning for instance confidence scores, costs, required clearance level, or number of `proof trees' have 
natural game-theoretic interpretations. In fact, provenance analysis of games is of independent interest, and provenance values
of positions in a game provide detailed information about the number and properties of the strategies of the players, far beyond 
the question whether or not a player has a winning strategy from a given position. 
\item We deal with negation by transformation to negation normal form. This is the common approach for 
the design of model checking games and game-based evaluation algorithms. But while this is there mainly a matter of convenience
(to avoid role switches between players during a play),  provenance semantics imposes even stronger reasons for transformations to negation normal form. 
Indeed, beyond Boolean semantics, negation is not a compositional logical operation:
the provenance value of $\neg\phi$ is not necessarily determined by the  provenance value of $\phi$.
\item On the algebraic side, we introduce new provenance semirings of polynomials and formal power series,
which take negation into account. They are obtained by taking quotients of traditional provenance semirings by congruences 
generated by products of positive and negative provenance tokens; they are called semirings of dual-indeterminate
polynomials or dual-indeterminate power series.  
\end{itemize}

Preliminary accounts of our approach, confined to first-order logic and without the connection to games,
but discussing potential applications to issues such as model updates, and 
reverse provenance analysis (e.g., confidence maximization), 
have been given in \cite{Tannen17} and \cite{GraedelTan17}. 
Here we put also the provenance analysis of games into focus, in fact we develop our approach here
from the perspectives of games. We shall first discuss the case of finite acyclic games
which are sufficient for the provenance analysis of first-order logic and its fragments. 
Most of the central issues of our approach, in particular the view of provenance values
in terms of valuations of strategies and plays, appear already in this simple scenario.
We shall then discuss reachability games on graphs that admit cycles. These are the games that
are relevant for the provenance analysis of logics with least  (but without greatest) fixed points. 
For these it will be necessary to restrict from arbitrary commutative semirings to
$\omega$-continuous ones. Such an analysis has previously been carried out for Datalog, but to
deal with (atomic) negation we have to combine this  with the idea of taking quotients by the duality on indeterminates,
which will lead us to semirings of dual-indeterminate power series.
Finally we shall outline a provenance approach for safety games and greatest fixed points. Our central algebraic tools here are
absorptive semirings, in particular the semiring $\Sinf[X]$ of generalized absorptive polynomials,
admitting also infinite exponents.

This paper is intended to lay foundations for our general approach to a provenance analysis 
of logic and games, that should take us far beyond the specific cases studied here. 
The application of the acyclic case to modal and guarded logics has been analysed in \cite{DannertGra19a}.
In \cite{XuZhaAlaTan18} our approach has been 
applied to database repairs; it has been shown how our treatment of negation, or absent information,
can be used to explain and repair missing query answers and the failure of integrity constraints in databases. 
Further, the potential of the provenance methods developed here for 
applications in knowledge represenation and description logics has been discussed in \cite{DannertGra19b}.
Work in progress includes the provenance analysis of temporal and dynamic logics in the setting of absorptive semirings, 
the study of logics of dependence and independence from the point of view of provenance,
and the algorithmic analysis of  computing provenance values in various settings.

\section{Commutative semirings}

\bdefn A \emph{commutative semiring}  is an algebraic structure 
$(K,+,\cdot,0,1)$, with $0\neq1$,  such that $(K,+,0)$
and $(K,\cdot,1)$ are commutative monoids, $\cdot$
distributes over $+$, and $0\cdot a=a\cdot0=0$. A semiring
is  \emph{+-positive}  if $a+b=0$ implies $a=0$ and
$b=0$. This excludes rings. A semiring is \emph{root-integral}
if $a\cdot a=0$ implies $a=0$.
All semirings considered in this paper are commutative,
+-positive and root-integral. 
Further, a commutative semiring is \emph{positive} if it is
+-positive and has no divisors of 0 (i.e. $a\cdot b=0$ implies
that $a=0$ and $b=0$). The standard semirings considered in
provenance analysis are in fact positive, but for an appropriate treatment 
of negation we shall introduce later in this paper semirings (of dual-indeterminate polynomials
or power series) that have divisors of 0.
\edefn

Notice that a semiring $K$ is positive if, and only if, 
the unique function $h: K\ra \{0,1\}$ with 
$h^{-1}(0)=\{0\}$ is a homomorphism from $K$ into the Boolean semiring $\B=(\{0,1\},\lor,\land,0,1)$. 
A semiring $K$ is (+)-idempotent if $a+a=a$, for all $a\in K$, and
$(+,\cdot)$-idempotent if, in addition, $a\cdot a=a$ for all $a$.
Further, $K$ is \emph{absorptive}
if $a+ab=a$, for all $a,b\in K$. Obviousy, every absorptive semiring is (+)-idempotent.

\medskip Elements of a commutative semiring will be used as truth values for
logical statements and as values for positions in games.
The intuition is that + describes the \emph{alternative use}
of information, as in disjunctions or existential quantifications, or for
different possible choices of a player in a game,
whereas $\cdot$ stands for the \emph{joint use} of information,
as in conjunctions or universal quantifications, or for choices in a game that are 
controlled by the opponent of the given player. Further, 0 is the value
of false statements or losing positions, whereas any element $a\neq 0$ of a semiring $K$ stands
for a ``nuanced''  interpretation of true or as a value of a non-losing position.   

\topic{Application semirings}
We briefly discuss some specific semirings that provide interesting
information but about a logical statement or a position in a game.
\begin{itemize}
\item The \emph{Boolean semiring} $\Bool=(\{0,1\},\vee,\wedge,0,1)$ is the standard habitat of 
logical truth. 
\item $\Nat=(\Nat,+,\cdot,0,1)$ is used here for counting winning strategies in games.
It also plays an important role for \emph{bag semantics} in databases.
\item $\Trop=(\mathbb{R}_{+}^{\infty},\min,+,\infty,0)$ 
is called the \emph{tropical} semiring. It has many applications
in several areas of computer science.
It is used here for measuring the cost of strategies.
\item The \emph{Viterbi} semiring $\Vit=([0,1],\max,\cdot,0,1)$
is isomorhic to $\Trop$ via $x\mapsto e^{-x}$ and $y\mapsto -\ln y$. 
We will think of the elements of $\Vit$ as \emph{confidence scores}
and use it to describe the confidence that a player can win from a given
position or the confidence assigned to a logical statement.
\item The \emph{min-max} semiring on a totally ordered set $(A,\leq)$
with least element $a$ and greatest  element $b$
is the semiring $(A,\max,\min,a,b)$. 
\end{itemize}

\topic{Provenance semirings}  Beyond the traditional application semirings, there
are some important provenance semirings of polynomials 
that are used for a general provenance analysis. These semirings have algebraic
\emph{universality} properties (they are freely generated) 
for various classes of semirings.
This allows us to compute provenance 
values once in a general such semiring and then to specialise it via homomorphisms 
(i.e. evaluation of the polynomials) to specific application semirings as needed. 

\begin{itemize}
\item For any set $X$, the semiring $\Nat[X]=(\Nat[X],+,\cdot,0,1)$
consists of the multivariate polynomials in indeterminates from $X$
and with coefficients from $\Nat$. This is the commutative
semiring freely generated by the set $X$. 
\item By dropping coefficients from $\Nat[X]$, we get the semiring
$\B[X]$ whose elements are just finite sets
of distinct monomials. It is the free (+)-idempotent semiring over $X$.
\item By dropping also exponents, we get the semiring $\W[X]$ of 
finite sums of monomials that are linear in each argument.
It is sometimes called the Why-semiring.
\item The free absorptive semiring $\S[X]$ over $X$ 
consists of 0,1 and all antichains of monomials
with respect to the component-wise order on their exponents.
It is the quotient of $\N[X]$ by the congruence induced by
$p\sim q$ for monomials $p,q$ with $p=qr$.  
\item Finally $\PosBool(X)=(\PosBool(X),\vee,\wedge,\bot,\top)$ is the semiring
whose elements are classes of equivalent positive (monotone)
boolean expressions with variables from $X$ (its elements
are in bijection with the positive boolean expressions 
in irredundant disjunctive normal form).
This is the distributive lattice freely generated by the set $X$.
\end{itemize}

\section{Games}

We consider two-player turn-based games on graphs. Such a game is defined by the game graph
on which it is played, and by the objectives of the players. 

\bdefn
 A \emph{game graph} is a structure
$\Gg=(V,V_0,V_1,T, E)$, where $V=V_0\cup V_1\cup T$ is the set of positions, partitioned into the sets
$V_0$, $V_1$ of the two players and the set $T$ of terminal positions, 
and  where $E\subseteq V\times V$ is the set of 
moves.
We denote the set of immediate successors of a position $v$ by 
$vE:=\{w: (v,w)\in E\}$ and require that $vE=\emptyset$ if, and only if, $v\in T$.
A play from an initial position $v_0$ is a finite or infinite path $v_0v_1v_2\dots$
through $\Gg$ where the successor $v_{i+1}\in v_iE$ is chosen by Player~0
if $v_i\in V_0$ and by Player~1 if $v_1\in V_1$.
A play ends when it reaches a terminal node $v_m\in T$.
\edefn

\bdefn For every game graph  $\Gg=(V,V_0,V_1,T, E)$, and every initial  position $v_0\in V$, the
\emph{tree unraveling} of $\Gg$ from $v_0$ is the
game tree $\Tt(\Gg,v_0)$ consisting of all finite paths from $v_0$.
More precisely, $\Tt(\Gg,v)=(V^\#,V_0^\#,V_1^\#,T^\#, E^\#)$,
where $V^\#$ is the set of all finite paths $\pi=v_0v_1\dots v_m$ through $\Gg$,
with $V_\sigma^\#=\{\pi v\in V^\#:  v\in V_\sigma\}$, $T^\#=\{\pi t\in V^\#:  t\in T\}$,
and $E^\#=\{(\pi v, \pi vv'): (v,v')\in E\}$. For most game-theoretic considerations, the
games played on $\Gg$ and its unravelings  are equivalent, via the canonical projection 
$\rho:\Tt(\Gg,v_0)\rightarrow \Gg$ that maps every path $\pi v$ to its end point $v$.   
\edefn

A strategy for a player in a game is a function that selects moves at points that are controlled
by that player. A strategy need not be defined at all positions of a player, but it must be
closed in the sense that it defines a move from each position that is reachable by
a play that is admitted by the strategy. There are several possibilities to define
the notion of a strategy formally. For our purposes it is convenient to
identify a strategy with the histories of plays that it admits, i.e. to view it 
as an appropriate subtree of  $\Tt(\Gg,v_0)$.

\bdefn\label{def:strategy} 
A \emph{strategy} of Player~$\sigma$ (for $\sigma\in\{0,1\}$)
from $v_0$
in a game $\Gg$  is a subtree of  $\Tt(\Gg,v_0)$,
of the form $\Ss=(W,F)$ with $W\subseteq V^\#$ and $F\subseteq (W\times W)\cap E^\#$, 
satisfying the following conditions:
\begin{itemize}
\item $W$ is closed under predecessors: if $\pi v\in W$ then also $\pi\in W$. 
\item If $\pi v\in W\cap V^\#_\sigma$, then $|(\pi v)F|=1$.
\item If $\pi v\in W\cap V^\#_{1-\sigma}$ then $(\pi v)F=(\pi v)E^\#$.
\end{itemize}
We write $\Strat_\sigma(v_0)$ for the set of all strategies of Player~$\sigma$
from $v_0$.
\edefn

In a strategy $\Ss=(W,F)$, the set $W$ is the part of $\Tt(\Gg,v_0)$ on which the strategy is defined, and
$F$ is the set of moves that are admitted by the strategy.
A strategy $\Ss\in\Strat_\sigma(v_0)$ induces the set $\Plays(\Ss)$ of those
plays from $v_0$ whose moves are consistent with $\Ss$. We call $\Ss$ well-founded if it
does not admit any infinite plays; this is always the case on finite acyclic 
game graphs, but need not be the case otherwise. 
The set of possible \emph{outcomes} of a strategy $\Ss$ is the set of terminal nodes
that are reachable by a play that is consistent with $\Ss$.
A strategy can
also be viewed as a function $\Ss: W\cap V_\sigma^\#\ra V$ such that $\Ss(\pi v)\in vE$
defines the node to which Player~$\sigma$ moves from $\pi v$.

\medskip
The simplest objectives of players are reachability and safety objectives.

\bdefn  A \emph{reachability objective} for Player~$\sigma$ is given by a set $T_\sigma\subseteq T$
of winning terminal positions. With such an objective, Player~$\sigma$ wins every play that reaches a position
in $T_\sigma$. Dually, a  \emph{safety objective} for Player~$\sigma$ is given by a set 
$L_\sigma\subseteq T$ of `losing' positions
that the player has to avoid, or equivalently, by its complement $S_\sigma=V\setminus L_\sigma$,
the region of safe positions inside of which the Player has to keep the play. With such an objective
Player~$\sigma$ wins every play, finite or infinite, that never reaches a position in $L_\sigma$.   
\edefn

Notice that the difference between reachability and safety objectives is relevant only
in cases where infinite plays are possible. 
Indeed, in a game that admits only finite plays, Player~$\sigma$ wins a play with the reachability objective 
$T_\sigma$ if, and only if, she wins that play with the safety objective given by $L_\sigma=T\setminus T_\sigma$,
so we can always reformulate reachability by safety and vice versa.
However, in a game that admits infinite plays, Player~$\sigma$ wins with a reachability objective 
$T_\sigma$ if, and only if, her opponent, Player~$1-\sigma$, loses with the safety condition $L_{1-\sigma}=T_\sigma$,
Hence winning with a reachability objective corresponds to defeating an opponent who plays with a safety objectives.
If both players play with reachability objectives, then infinite plays are won by neither player.

\section{Provenance for well-founded  games}\label{sect:wfgames}

We first study the provenance analysis of games for 
well-founded games, i.e. games that are played on \emph{finite acyclic 
game graphs} 
$\Gg=(V,V_0,V_1,T,E)$, and hence
do not admit infinite plays. 
We introduce \emph{$K$-valuations} $f_0$ and $f_1$ that associate with
every position $v\in V$ provenance values $f_0(v)$ and $f_1(v)$, respectively. 
The idea is that, for $\sigma\in\{0,1\}$, the function $f_\sigma$
describes the value of each position from the point if view of Player~$\sigma$.  
Such a valuation is induced by its values on the terminal positions,
i.e. by a function $f_\sigma:T\ra K$, and by a valuation of the moves,
i.e. by a function $h_\sigma:E\ra K\setminus\{0\}$.
Here, the function $f_\sigma:T\ra K$ 
defines the value, for Player~$\sigma$, of every \emph{terminal position} where,
intuitively, $f_\sigma(t)=0$ means that position $t$ is losing for Player $\sigma$.
In the simplest case, we can specify reachability objectives $T_\sigma$ by 
setting $f_\sigma(t)=1$ for $t\in T_\sigma$ and $f_\sigma(t)=0$ otherwise.
The functions $h_\sigma:E\ra K\setminus\{0\}$ 
provide a value (or cost) for Player~$\sigma$ of the moves. 
In many cases valuations of moves are not relevant; we then just put $h_\sigma(vw)=1$ for all 
edges $(v,w)\in E$.

The extension of the basic valuations $f_\sigma:T\ra K$ and $h_\sigma:E\ra K\setminus\{0\}$
to valuations $f_\sigma:V\ra K$ for all positions then relies on the idea that 
a move from $v$ to $w$ contributes to $f_\sigma(v)$ the  value
$h_\sigma(vw)\cdot f_\sigma(w)$. These contributions are summed up
in the case that $v$ is a position for Player~$\sigma$ (i.e. when she choses herself the successors), 
and  multiplied in the case that $v$ is a position of the opponent (i.e. when she has to
cope with any of the possible successors).  This is summarized by the following 
definition.

\begin{definition} \label{defn: acyclic K-valuation}
Let $K$ be a commutative semiring, let
$\Gg=(V,V_0,V_1,T,E)$ be a finite acyclic game graph, and let $\sigma\in\{0,1\}$
denote one of the two players.
A \emph{$K$-valuation} of $\Gg$ for Player~$\sigma$ is a function $f_\sigma:V\ra K$.
It is defined from basic valuations $f_\sigma:T\ra K$ and $h_\sigma:E\ra K\setminus\{0\}$
via backwards induction, by
\[   f_\sigma(v):=\begin{cases}
\sum_{w\in vE} h_\sigma(vw)\cdot f_\sigma(w)&\text{ if }v\in V_\sigma\\
\prod_{w\in vE} h_\sigma(vw)\cdot f_\sigma(w)&\text{ if }v\in V_{1-\sigma}.\end{cases}
\]
\end{definition}

An equivalent characterization of the $K$-valuation $f_\sigma$
can be obtained by defining provenance values for plays and strategies.

\begin{definition}
For a play $x=v_0v_1\dots v_m$ from $v_0$ to a terminal node $v_m$, we 
define its valuation for Player~$\sigma$ as
$f_\sigma(x):=h_\sigma(v_0v_1)\cdots h_\sigma(v_{m-1}v_m)\cdot f_\sigma(v_m)$. 
Let now $\Ss=(W,F)\subseteq \Tt(\Gg,v_0)$ be a strategy for Player~$\sigma$ from $v_0$ and 
$\rho_S: (W,F)\ra (V,E)$ be the restriction of of the canonical homomorphism
$\rho:\Tt(\Gg,v_0)\ra \Gg$ to $\Ss$. For any position $v\in V$ and
any move $e\in E$, the values
\[   \#_\Ss(v):=| \rho_\Ss^{-1}(v)|\quad\text{and}\quad \#_\Ss(e):=| \rho_\Ss^{-1}(e)|\]
indicate how often the position $v$ and the move $e$ appear in the strategy $\Ss$.
We then define the provenance value $\Ss\in\Strat_\sigma(v_0)$ as 
\[     F(\Ss):= \prod_{e\in E} h_\sigma(e)^{\#_\Ss(e)} \cdot \prod_{v\in T} f_\sigma(v)^{\#_\Ss(v)}.\]
\end{definition}

In some important special cases, provenance values of  strategies coincides with
the product of the provenance values over all plays that they admit. 

\begin{lemma}\label{strategy-value}  If $h_\sigma(e)=1$ for all moves $e\in E$, or if the underlying
semiring is multiplicatively idempotent (i.e. $a^2=a$ for all $a$) we have that  
$F(\Ss)= \prod_{x\in\Plays(\Ss)}   f_\sigma(x)$ for all $\Ss\in\Strat_\sigma(v)$.
\end{lemma}

However, there are simple games where this is not the case. Consider, for instance, the valuation for
Player~0 in a game where only the opponent, Player~1, moves: from position $v$,
Player ~1 can proceed to $w$ by a move with value $h_0(vw)=a$, and from $w$
he has the choice of moving to either $s$ or $t$, both options having value 1 for Player 0.
There is only one strategy $\Ss$ for Player 0 (do nothing), with provenance value $a$.
However, the strategy admits two plays, ending in $s$ and $t$, respectively, 
both of which have value $a$. Thus the product over
the provenance value of the plays is $a^2$.

\begin{theorem}\label{thm:prov-games}
For any commutative semiring $K$ and any finite acyclic game $\Gg$, 
let $f_\sigma:V\ra K$ be the provenance valuation 
for Player~$\sigma$, induced by the valuation $f_\sigma:T\ra K$ of the terminal nodes and $h_\sigma: E\ra K\setminus\{0\}$
of the moves. Then, for every position $v$
\[    f_\sigma(v)   =\sum_{\Ss\in\Strat_\sigma(v)}\  F(\Ss).\]
\end{theorem}

\lueg  For terminal positions $v$ the claim is trivially true. So suppose that 
$v\in V_\sigma$. Then any strategy 
$\Ss\in\Strat_\sigma(v)$ can be written in the form $\Ss=v\cdot\Ss'$ for some successor $w\in vE$
and some strategy $\Ss'\in\Strat_\sigma(w)$. 
Clearly,  $\#_\Ss(t)=\#_{\Ss'}(t)$ for every terminal position $t\in T$.
For the moves we have that $\#_\Ss(e)=\#_{\Ss'}(e)$ for all $e\neq (v,w)$ but
$\#_S(e)=1$ and $\#_{\Ss'}(e)=0$ for $e=(v,w)$. 
This implies that $F(S)=h(vw)\cdot F(\Ss')$.
By induction hypothesis 
$f_\sigma(w)=\sum_{\Ss'\in\Strat_\sigma(w)}\ F(\Ss')$. Hence 
\[  f_\sigma(v)=\sum_{w\in vE} h_\sigma(vw)\cdot f_\sigma(w)=
\sum_{w\in vE} \sum_{\Ss'\in\Strat_\sigma(w)}  h_\sigma(vw)\cdot F(\Ss')
=  \sum_{\Ss\in\Strat_\sigma(v)}  F(\Ss).
\]
Finally, let $v\in V_{1-\sigma}$ with $vE=\{w_1,\dots,w_n\}$.  Every strategy $\Ss\in\Strat_\sigma(v)$ has the form 
$\Ss=v(\Ss_1\cup\dots\cup \Ss_n)$ with $\Ss_i\in\Strat_\sigma(w_i)$.
For the terminal nodes $t\in T$ we have that $\#_\Ss(t)=\sum_{i\leq n}\#_{\Ss_i}(t)$; similarly, for all moves $e$ 
from a different position than $v$, we have $\#_\Ss(e)=\sum_{i\leq n}\#_{\Ss_i}(e)$, but for the
moves $e=(v,w_i)$ we have $\#_\Ss(e)=1$ and $\#_{\Ss_i}(e)=0$ for all $i$.
Thus $F(\Ss)= \prod_{w_i\in vE}  h_\sigma(vw_i)\cdot F(\Ss_i))$.
 It follows that 
\begin{align*}   f_\sigma(v)&=\prod_{w_i\in vE} h_\sigma(vw_i)\cdot f_\sigma(w_i)=
\prod_{w_i\in vE} h_\sigma(vw_i)\cdot \sum_{\Ss_i\in\Strat_\sigma(w_i)} F(\Ss_i)\\
&=  \sum_{v\cdot(\Ss_1\cup\dots \Ss_n)\in\Strat_\sigma(v)} \prod_{w_i\in vE} h_\sigma(vw_i)\cdot F(\Ss_i)) =
\sum_{\Ss\in\Strat_\sigma(v)}\    F(\Ss).
\end{align*}
\hx

From this description, we can derive a number of applications of provenance 
valuations on games.
We first consider the information provided by valuations in the general
provenance semirings of polynomials.
Let $\N[T]$ be the semiring of polynomials with coefficients in $\N$
over indeterminates $t\in T$, where $T$ is the set of terminal positions
in an acyclic  game graph $\Gg=(V,V_0,V_1,T,E)$.
Let $f_\sigma:V\ra\N[T]$ be the  valuation induced by
setting $f_\sigma(t)=t$ for $t\in T$ and $h_\sigma(vw)=1$ for all edges $(v,w)$,
so that the value of a play is just its outcome, i.e. the terminal position
where it ends.

Clearly, we can write $f_\sigma(v)$ as a sum of monomials 
$m\cdot t_1^{j_1}\dots t_k^{j_k}$. This provides a detailed description of
the number and properties of the strategies that Player~$\sigma$ has from
position $v$.

\begin{theorem} The valuation $f_\sigma(v)\in\N[T]$ is the sum of
those monomials $m\cdot t_1^{j_1}\dots t_k^{j_k}$ (with $m,j_1\dots,j_k>0$)  
such that Player~$\sigma$ has precisely $m$ strategies $\Ss\in\Strat_\sigma(v)$
with the property that the set of possible outcomes for $\Ss$ is precisely $\{t_1,\dots,t_k\}$, and 
precisely $j_i$ plays that are consistent with $\Ss$ have the outcome $t_i$.
\end{theorem}

This is an immediate consequence of Theorem~\ref{thm:prov-games} and Lemma~\ref{strategy-value}.
In many cases, somewhat less detailed information is sufficient, which can
be obtained by valuations in less informative provenance semirings than $\N[T]$:
\begin{itemize}
\item Evaluating $f_\sigma(v)$ in the idempotent semiring $\B[T]$ gives us the sum
of monomials $t_1^{j_1}\dots t_k^{j_k}$ for which Player~$\sigma$ has
at least one strategy whose multiset of admitted outcomes consists
of $t_1,\dots,t_k$ with multiplicities $j_1,\dots,j_k$, respectively.
\item If we evaluate $f_\sigma(v)$ in $\W[T]$ we get the sum of
monomials $t_1\dots t_m$ such that Player~$\sigma$ has
a strategy whose set of outcomes is $\{t_1,\dots t_m\}$. The information
on multiplicities of strategies and outcomes is dropped.
\item An interesting case is the evaluation in the absorptive semiring
$\S[X]$.  For two strategies $\Ss,\Ss'\in\Strat_\sigma(v)$, we say that 
$\Ss$ \emph{absorbs} $\Ss'$ if  for every 
terminal position $t\in T$, $\Ss$ admits less plays with outcome $t$
than $\Ss'$. We call $\Ss$ \emph{absorption-dominant} if it is
not absorbed by any other strategy. Now, $f_\sigma(v)\in\S[X]$ is
the sum of monomials $t_1^{j_1}\dots t_k^{j_k}$ that describe 
precisely the (multiset of outcomes of the) absorption dominant strategies
of Player~$\sigma$ from $v$. See Sect.~\ref{sect:strategy-absorption} below
for a more detailed analysis of absorption among strategies.
\item Finally, the evaluation of $f_\sigma(v)\in\PosBool[T]$ consists of
those monomials $t_1\dots t_k$ such that $\{t_1,\dots,t_k\}$
a minimal set among the sets of outcomes of strategies $\Ss\in\Strat_\sigma(v)$. 
\end{itemize}

Fix any reachability objective $W\subseteq T$. In any of these provenance semirings,
we can write the polynomial $f_\sigma(v)$ as a sum $f_\sigma(v)=f^W_\sigma(v)+g^W_\sigma(v)$ where 
$f^W_\sigma(v)$ is the sum of those monomials that only contain indeterminates in $W$ 
and $g^W_\sigma(v)$ contains the rest. 

\begin{theorem} For every subset $W\subseteq T$ and every $v\in V$, Player $\sigma$ has a strategy to
reach $W$ from $v$ if, and only if, $f^W_\sigma(v)\neq 0$ (in any of the provenance semirings given above). 
Moreover, if  we set $f(t)=1$ for $t\in W$ and $f(t)=0$ for $t\in T\setminus W$, and evaluate $f_\sigma$ in the semiring 
$\N$ of natural numbers,  then $f_\sigma(v)$ 
is the number of distinct winning strategies for Player $\sigma$ to reach $W$ from $v$.
\etheo  

Evaluation in other application semirings gives further interesting information about strategies:

\medskip\noindent{\bf Cost of strategies. }
Given a game $\Gg$, we associate with Player~0 \emph{cost functions}
$f_0: T \ra  \R_+$ and $h:E\ra \R_+$ for the terminal positions and the moves.  
We define the cost of a strategy $\Ss\in\Strat_0(v)$ as the sum of the costs of all 
moves and outcomes  that it admits, weighted by the number of their occurrences.

\begin{proposition}
The \emph{cost of an optimal strategy} from $v$ in $\Gg$ is
given by the valuation $f_0(v)$ in the tropical semiring $\Trop=(\R_+^\infty,\min,+,\infty,0)$.
\end{proposition}

\lueg  Since the product in $\Trop$ is addition in $\R_+^\infty$, the cost of a strategy $\Ss$
for Player~0, as defined above, coincides with the valuation $f_0(\Ss)$ in $\Trop$.
The summation in $\Trop$ is minimization in $\R_+^\infty$, so from 
Theorem~\ref{thm:prov-games} we get that
\[   f_0(v)= \min_{\Ss\in\Strat_0(v)} \   F(\Ss)\]
describes indeed the minimial cost of a strategy for Player~0 from position $v$. 
\hx

\medskip\noindent{\bf Clearance levels. } The access control semiring is 
$\Access=(\{\Pub<\Cnf<\Sec<\Tsec<0\},\min,\max,0,\Pub)$
where $\Pub$ is ``public'', $\Cnf$ is ``confidential'',
$\Sec$ is ``secret'', $\Tsec$ is ``top secret'', and $0$ is ``so secret that nobody can access it!''.
Let  $f_\sigma: T \ra  \Access$ and  $h_\sigma :E\ra \Access\setminus\{0\}$ define access
levels for the terminal positions and the moves for Player~$\sigma$, in the sense that
Player~$\sigma$ can make a move $e$ if, and only if, his personal clearance level is at least $h(e)$
and similarly, he can access a terminal position $t$ if, and only if, his clearance level is at least $f_\sigma(t)$.   

\begin{proposition}
The valuation $f_\sigma(v)\in \Access$ describes the \emph{minimal clearance level} that Player~0 needs
to win from position $v$, i.e. to have a strategy that guarantees to reach a terminal position that is accessible 
for him.
\end{proposition}

The proof is a straightforward induction.

\medskip\noindent{\bf Confidence in games. }
Suppose that $f_\sigma: T\ra [0,1]$ describes the confidence that Player~$\sigma$ puts into $t$ being a 
winning position for her. We want to compute \emph{confidence scores} $f_\sigma(v)$ to describe the
confidence of Player~$\sigma$ that she can win from $v$. It is natural to define the confidence score $f_\sigma(v)$
as the \emph{maximum} of the confidence scores of the successors $w\in vE$ 
in the case that $v\in V_\sigma$.
For confidence scores of combinations of events whose choice is taken by an opponent,
such as for the possible moves from a position $v\in V_{1-\sigma}$,  there
are different approaches in the literature. A popular one, with which we work here,  
takes the \emph{product} of the confidence scores of the events from which the opponent
choses. Adopting this definition, the following proposition is immediate.

\begin{proposition}
Confidence scores are computed as semiring valuations $f_\sigma: V\ra \V$ in the
Viterbi semiring $\V=([0,1], \max,\,\cdot\, ,0,1)$.
\end{proposition}

\medskip\noindent{\bf Min-Max Games. }
Finally note that valuations in a min-max semiring $(A,\max,\min,a,b)$
describe the value of positions in games 
where Player 0 tries to maximize and Player~1 tries to minimize the outcome of the play.

\medskip\noindent{\bf Separating Valuations. }
The $K$-valuations $f_0,f_1$ for the two players in a game $\Gg$, as defined 
by Definition~\ref{defn: acyclic K-valuation}, are \emph{a priori} completely independent of each other. 
This admits the treatment of a wide variety of games, without
any restrictions on how the objectives of the two players relate to each other.
For instance, in a completely cooperative game, the basic valuations of of the terminal positions
would be the same for Player~0 and Player~1. However, in many games, the objectives of the two players
are antagonistic, and valuations $f_0$ and $f_1$ should reflect this. This motivates the following definition.

\bdefn Let $\Gg$ be a game graph, with valuations $f_0,f_1$ for the two players in a semiring $K$, 
and let $U\subseteq V$ be a set of positions.
We say that 
\begin{enum}
\item $f_0,f_1$ for the two players are \emph{separating} on $U$ if for all $u\in U$, either $f_0(u)=0$ or  $f_1(u)=0$. 
\item $f_0, f_1$ are \emph{weakly separating} on $U$ if $f_0(u)f_1(u)=0$ for all $u\in U$. Notice that in the case where $K$ has no 
divisors of 0, weakly separating valuations are in fact separating.
\item $f_0$ and $f_1$ are \emph{strongly separating} on $U$, if they are separating, and in addition, 
$f_0(u)+f_1(u)\neq 0$ for all $u\in U$.
\end{enum}
\edefn

\prop If two valuations $f_0$ and $f_1$ are (weakly) separating on the the terminal positions of $\Gg$,
then they are (weakly) separating on all positions of $\Gg$.
\eprop

\lueg Recall that all our semirings are assumed to be +-positive. For $v\in V_\sigma$, we have that
\[    f_\sigma(v)=\sum_{w\in vE} h(vw)f_\sigma(w) \text{ and } f_{1-\sigma}(v) =\prod_{w\in vE} h(vw)f_{1-\sigma}(w).\]
It follows that $f_0$ and $f_1$ are separating on $v$ if they are separating on all $w\in vE$. Further,
\begin{align*}
f_\sigma(v)f_{1-\sigma}(v)=&\Bigl(\sum_{w\in vE} h_\sigma(vw)f_\sigma(w)\Bigr)\Bigl(\prod_{w\in vE} h_{1-\sigma}(vw)f_{1-\sigma}(w)\Bigr)=\\
&\sum_{w\in vE}\Bigl( h_\sigma(vw)f_\sigma(w)\prod_{w'\in vE} h_{1-\sigma}(vw')f_{1-\sigma}(w')\Bigr)=\\
&\sum_{w\in vE}\bigl( h_\sigma(vw)h_{1-\sigma}(vw)f_\sigma(w)f_{1-\sigma}(w)\prod_{w'\in vE\setminus\{w\}} h_{1-\sigma}(vw') f_{1-\sigma}(w')\bigr).\end{align*}
This proves that $f_0$ and $f_1$ are weakly separating on $v$ if they are so on all $w\in vE$.
\hx

The corresponding implication for strongly separating valuations does not hold for all +-positive semirings,
but it holds for positive ones.

\prop If two valuations $f_0$ and $f_1$ into a positive semiring are strongly separating on the the terminal positions of $\Gg$,
then they are so on all positions of $\Gg$.
\eprop

\lueg
By induction. Assume that $f_0$ and $f_1$ are strongly separating on all $w\in vE$.
Then $f_\sigma(v)+f_{1-\sigma}(v)=0$ only if $f_\sigma(w)=0$ for all $w\in vE$ and $f_{1-\sigma}(w)=0$
for at least one $w\in vE$. But this implies that $f_0(w)+f_1(w)=0$ for some $w\in vE$ 
which contradicts our assumption.
\hx
 
Note that for the Boolean semiring $K=\B$, this is just Zermelo's Theorem on the determinacy of reachability games on well-founded
game graphs: from every position, one of the two players has a winning strategy.

\medskip\noindent{\bf Counting positional winning strategies? }
A strategy is \emph{positional} if it only depends on the 
current position, and not on the history of the play, i.e. if $\Ss(\pi v)=\Ss(\pi' v)$
for all $v$ and all paths $\pi v$, $\pi'v$ that lead to $v$. A positional 
strategy can be described by a function $s: V_\sigma\ra V$ or by a subgraph
$\Ss$ of $\Gg$ (rather than of $\Tt(\Gg,v_0)$).

Given that in the study of games there is (for instance for algorithmic reasons)
a strong interest in positional strategies, it is reasonable to ask whether there
exist valuations in different semirings that count just the positional strategies.
However, invariance under counting bisimulation shows that this is not possible.

\bdefn
Let $\Gg=(V,V_0,V_1,T, E)$ and $\Gg'=(V',V'_0,V'_1,T', E')$  be two game graphs.
A \emph{counting bisimulation} between $\Gg$ and $\Gg'$ is a relation $Z\subseteq V\times V$
such that for every pair $(v,v')\in Z$ we have that
\begin{enum}
\item $v\in V_\sigma$ if, and only if, $v'\in V'_\sigma$ and $v\in T$ if, and only if, $v'\in T'$, and
\item there is a local bijection  $z_{vv'}: vE \ra v'E'$ between the
immediate successors of $v$ and $v'$ such that $(w,z_{vv'}(w))\in Z$, for every $w\in vE$.
\end{enum}
We write $\Gg,v \sim \Gg',v'$ if there is a counting bisimulation $Z$ between $\Gg$ and $\Gg'$
such that $(v,v')\in Z$. Notice that for any game graph $\Gg$, the relation
$Z=\{(v,\pi v): v\in V, \pi v\in V^\#\}$ is a counting
bisimulation between $\Gg$ and its unraveling $\Tt(\Gg,v_0)$.
\edefn

$K$-valuations of games are invariant under counting bisimilarity in the following sense.
Let $\Gg$ and $\Gg'$ be two acyclic game graphs with $K$-valuations $f_\sigma:T\ra K$ and 
$f'_\sigma: T'\ra K$ of the terminal positions and $h:E\ra K$ and $h':E'\ra K$ of the 
moves. We say that a counting bisimulation $Z\subseteq V\times V'$ respects these
valuations if $f_\sigma(t)=f'_\sigma(t')$ for all $(t,t')\in Z\cap T\times T'$, and
$h_\sigma(vw)=h'_\sigma(v'w')$ whenever $(v,v')\in Z$ and $(w,w')\in Z$.

\prop Let $Z$ be a counting bisimulation between $\Gg$ and $\Gg'$ that
respects the basic valuations of the terminal positions and the moves.
Then $Z$ respects the valuations of all positions, i.e.
$f_\sigma(v)=f'_\sigma(v')$ for all $(v,v')\in Z$.
\eprop

\lueg Let $(v,v')\in Z$. If $v$ and $v'$ are terminal positions, then $f_\sigma(v)=f'_\sigma(v')$ by
assumption. Otherwise, $v$ and $v'$ are both positions of the same player.
If they belong to Player~$\sigma$, then $f_\sigma(v)=\sum_{w\in vE} f_\sigma(w)$.
The local bijection $z_{vv'}$ maps every $w\in vE$ to some $w'\in v'E'$ such that,
by induction hypothesis, $f_\sigma(w)=f'_\sigma(w')$. Hence 
$f'_\sigma(v')=\sum_{w'\in v'E'} f'_\sigma(w')=\sum_{w\in vE} f_\sigma(w)=f_\sigma(v)$.
If $v$ and $v'$ belong to Player~($1-\sigma$) the reasoning is completely
analogous, taking a product rather than a sum.
\hx

In particular $K$-valuations of acyclic games do not change if we replace
a game graph $\Gg$ by one of its unravelings $\Tt(\Gg,v)$.
Indeed, every valuation $f_\sigma:T\ra K$ on the terminal positions of
a game graph $\Gg$ extends to the same valuation for $v$ on $\Gg$ as on
the tree unraveling $\Tt(\Gg,v)$.
On the other side, every strategy on a tree-shaped game graph is positional.
Thus the number of positional winning strategies is certainly not invariant
under unraveling and hence not definable by valuations in a semiring.

\section{Provenance for first-order logic via model checking games and dual-indeterminate polynomials}

Given a finite relational vocabulary $\tau$  and
a finite  non-empty universe $A$, we denote by $\Atoms_A(\tau)$ the set of all 
atoms $R\bar a$ with $R\in \tau$ and $\bar a\in A^k$.
Further, let $\NegAtoms_A(\tau)$ be the set of all negated atoms
$\neg R\bar a$ where $R\bar a\in\Atoms_A(\tau)$, and consider the set of all $\tau$-literals on $A$,
\[ \Lit_A(\tau):= \Atoms_A(\tau)\cup \NegAtoms_A(\tau)\cup\{ a \op b :  a, b \in A\},\]
where $\op$ stands for $=$ or $\neq$.

\bdefn Given any commutative semiring $K$, a \emph{$K$-interpretation} (for $\tau$ and $A$) is
a function $\pi: \Lit_A(\tau)\ra K$ that maps equalities and inequalities to their truth values 0 or 1.
\edefn

We have defined in \cite{GraedelTan17} how a semiring interpretation extends to a full valuation 
$\pi: \FO(\tau)\ra K$ mapping any fully instantiated formula $\psi(\bar a)$ (or equivalently, any first-order 
sentence of vocabulary $\tau\cup A$),  to a value $\pi\ext{\psi}$, by setting
\begin{align*}
&\pi\ext{\psi\lor\phi} := \pi\ext{\psi}+\pi\ext{\phi)} \qquad 
&&\pi\ext{\psi\land\phi} := \pi\ext{\psi} \cdot \pi\ext{\phi}\\
&\pi\ext{\E x \phi(x)}:=\sum_{a\in A} \pi\ext{\phi(a)}\qquad &&\pi\ext{\A x \phi(x)}:=\prod_{a\in A} \pi\ext{\phi(a)}.
\end{align*}
Negation is handled via negation normal forms: we set
$\pi\ext{\neg \phi}:=\pi\ext{\nnf(\neg \phi)}$ where $\nnf(\phi)$ is the negation normal form of $\phi$.
   
This is equivalent to the game provenance, as defined above, for the model checking
game associated with the formula $\psi$ and the $K$-interpretation $\pi:\Lit_A(\tau)\ra K$.
Notice that classically, model checking games are defined for a formula 
(assumed to be given in negation normal form) and a fixed structure $\AA$
(see e.g. \cite[Chap.~4]{AptGraedel11}).
However, the game graph of such a model checking game depends only on
the formula $\psi$ and the \emph{universe} $A$ of the given structure $\AA$.
It is only the labelling of the terminal positions of the game, as winning 
for either the Verifier (Player~0) or the Falsifier (Player~1), that depends
on which of the literals in $\Lit_A(\tau)$ are true in $\AA$.
Hence the definition of a model checking game readily generalizes 
to our more abstract provenance scenario. 

\bdefn\label{FOgame} Let $\psi(\bar x)\in\FO(\tau)$ be a first-order formula in negation normal form
with a relational vocabulary $\tau$, and let $A$ be a (finite) universe.
The model checking game $\Gg(A,\psi)$ has
positions $\phi(\bar a)$, obtained from a subformula $\phi(\bar x)$ of $\psi$, 
by instantiating the free variables $\bar x$ by a tuple $\bar a$ of elements of $A$. 
At a disjunction $(\psi\lor\phi)$, Player~0 (Verifier) moves to either $\psi$ or $\phi$,
and at a conjunction, Player~1 (Falsifier) makes an analogous move.
At a position $\E x\phi(\bar a,x)$, Verifier selects an element $b$ and moves
to $\phi(\bar a,b)$, whereas at positions $\A x \phi(\bar a,x)$ the move to
to the next position $\phi(\bar a,b)$ is done by Falsifier.  
The terminal positions of $\Gg(A,\psi)$ are the literals in $\Lit_A(\tau)$.
\edefn

A $K$-interpretation $\pi:\Lit_A(\tau)\ra K$ thus provides a valuation
of the set $T\subseteq\Lit_A(\tau)$ of terminal positions of the model checking game $\Gg(A,\psi)$,
for any sentence $\psi\in\FO(\tau\cup A)$. We view it as a valuation
$f_0$ for Player~0. The associated valuation $f_1$ for Player 1
is obtained by setting $f_1(\phi)=\pi\ext{\neg\phi}$ for
any literal $\phi\in\Lit_A(\tau)$.
Both valuations then extend to full valuations $f_0$ and $f_1$ of all positions
of $\Gg(A,\psi)$, including the position $\psi$ itself. The following result
is proved by a straightforward induction on formulae.

\begin{theorem} For all positions $\phi$ of $\Gg(A,\psi)$ we have that 
$f_0(\phi)=\pi\ext{\phi}$ and $f_1(\phi)=\pi\ext{\neg\phi}$.
\end{theorem}

Although this theorem holds without any restrictions on the semiring $K$ and the
$K$-interpretation $\pi$, not all such $K$-interpretations are really meaningful
for logic. Indeed the provenance value of complementary literals $R\bar a$ and
$\neg R\bar a$ have to be related in a reasonable way, and as a consequence 
also the general provenance semirings of polynomials need to be modified. 
In the simplest case a $K$-interpretation defines a unique $\tau$-structure.

\bdefn  A semiring interpretation $\pi: \Lit_A(\tau)\ra K$ is \emph{model-defining} if
for every atom $\phi\in\Atoms_A(\tau)$ one of $\pi(\phi)$ and $\pi(\neg\phi)$
is 0, and the other is $\neq 0$. It uniquely defines the $\tau$-structure $\AA_\pi$ that has
universe $A$, and in which precisely those literals $\phi$ are true for which $\pi(\phi)\neq 0$. 
\edefn

Notice that if $K$ is not the Boolean semiring, then several different $K$-interpretations may define
the same structure. Further, $K$-interpretations are interesting, and have a number of applications, also
in cases where they do not specify a single model, see \cite{GraedelTan17} and the references given there.

\topic{Dual-Indeterminate Polynomials}
Let $X,\nnX$ be two disjoint sets together with a one-to-one
correspondence $X\leftrightarrow\nnX$. We denote by $p\in X$ and
$\nnp\in\nnX$ two elements that are in this correspondence.  We refer
to the elements of $X\cup\nnX$ as \emph{provenance tokens} 
and we shall use ``positive''  and ``negative" tokens $p$ and $\nnp$ 
to annotate atoms $R\bar a\in\Atoms_A(\tau)$ and 
negated atoms $\neg R\bar a\in\NegAtoms_A(\tau)$, respectively.
By convention, if we annotate $R(\bar a)$ with $p$ 
then the ``negative'' token $\nnp$ can only be
used to annotate $\neg R(\bar a)$, and vice versa. 
We refer to $p$ and $\nnp$ as \emph{complementary} tokens.

\bdefn \label{dual-indeterminate} The semiring $\N[X,\nnX]$ is the quotient
of the semiring of polynomials $\N[X\cup\nnX]$ by the
congruence generated by the equalities $p\cdot\nnp=0$ for
all $p\in X$. This is the same as quotienting by the ideal generated 
by the polynomials $p\nnp$ for all $p\in X$.
Observe that two polynomials $g,g'\in\N[X\cup \nnX]$
are congruent if, and only if, they become identical after deleting from each of them
the monomials that contain complementary tokens. Hence, the 
congruence classes in $\N[X,\nnX]$ are in one-to-one correspondence
with the polynomials in $\N[X\cup\nnX]$ such that none of their monomials contain 
complementary tokens. We shall call these \emph{dual-indeterminate polynomials}.
\edefn

Note that $\N[X,\nnX]$ is $+$-positive and root-integral, but not positive,
since it has divisors of $0$. Further, we have the following \emph{universality property}: 
\begin{proposition}
\label{prop:prov-univ}
Every function $f:X\cup\nnX\rightarrow K$ into any commutative
semiring $K$ with the property that 
$f(p)\cdot f(\nnp)=0$ for all $p\in X$ extends uniquely
to a semiring homomorphism 
$h:\N[X,\nnX]\rightarrow K$ that coincides with $f$ on $X\cup\nnX$.
\end{proposition}

\bdefn
A \emph{provenance-tracking} interpretation is a mapping 
$\pi:\Lit_A(\tau)\rightarrow X\cup\nnX\cup\{0,1\}$ such that
$\pi(\Atoms_A(\tau))\subseteq X\cup\{0,1\}$
and $\pi(\NegAtoms_A(\tau))\subseteq\nnX\cup\{0,1\}$. Further,
$\pi$ maps equalities and inequalities to their truth values 0 or 1.
\edefn

The idea is that if $\pi$ annotates a positive or negative atom
with a token, then we wish to track that literal through the model-checking
computation. On the other hand annotating with $0$ or $1$ is done when
we do not track the literal, yet we need to recall whether it holds or not
in the model. See \cite{GraedelTan17} for more details and potential applications
of provenance-tracking interpretations.

\section{Semirings of dual-indeterminate power series and least fixed point solutions}

It is known that the general properties of commutative semirings  are not sufficient to deal with
unbounded iterations as they occur in fixed-point logic. Even for Datalog, one of the simplest
fixed-point formalism that omits the complications arising with universal quantification and
negation, appropriate semirings have the additional property of being \emph{$\omega$-continuous}.
The general $\omega$-continuous provenance semirings are no longer semirings of polynomials,
but semirings of formal power series, such as $\N^\infty[\![X]\!]$. We combine this here
with our approach for dealing with  negation by taking quotients with respect to
the congruence generated by products $p\bar p$ of positive and negative provenance tokens.  
What we obtain are $\omega$-continuous provenance semirings of dual-indeterminate power series,
such as $\N^\infty[\![X,\bar X]\!]$, as well as idempotent, absorptive, and other variants thereof. 
 
\medskip
A semiring $K$ is \emph{naturally ordered} if the relation $a\leq b:\Leftrightarrow \E x (a+x=b)$
is a partial order.  Note that  this relation is reflexive and
transitive in every semiring, but it is not always antisymmetric.
An $\omega$-chain is a sequence $(a_i)_{i\in\omega}$ with $a_i\leq a_{i+1}$ for all $i\in\omega$.

\bdefn\label{def:omega-cont} A commutative semiring $K$ is \emph{$\omega$-continuous} if it is naturally ordered and satisfies the
following additional conditions:
\begin{itemize}
 \item Every $\omega$-chain $(a_i)_i\in\omega$ has a supremum $\sup_{i\in\omega} a_i$  in $K$.
 As a consequence, we have a well-defined infinite summation operator $\sum$, such that for every sequence
 $(b_i)_{i\in\omega}$,
 \[  \sum_{i\in\omega} b_i :=  \sup \{a_0+\cdots +a_n:  n\in \omega\}\]     
 \item For every sequence $(a_i)_{i\in\omega}$ in $K$, every $c\in K$, and every partition $(I_j)_{j\in J}$ of $\omega$,
 we have that $c\cdot \sum_{i\in\omega} a_i = \sum_{i\in\omega} c\cdot a_i$ and 
 $\sum_{j\in J} \sum_{i\in I_j} a_i= \sum_{i\in\omega} a_i$.
 \end{itemize}
In an $\omega$-continuous semiring we further have the Kleene star operation, 
$a^*:=\sum_{i\in\omega} a^i=\sup_{i\in\omega} (1+a+a^2+\cdots +a^i)$.
A function $f:K\ra K$ is $\omega$-continuous if, $\sup_{i\in\omega} f(a_i)= f (\sup_{i\in\omega} a_i)$
for every $\omega$-chain $(a_i)_{i\in\omega}$. A consequence of the definition
is that any function defined by a polynomial or a power series is 
$\omega$-continuous in each argument.
\edefn

\bdefn Given a semiring $K$ and a finite set $X$ of indeterminates, we denote by $K[\![X]\!]$
the semiring of formal power series (i.e. possibly infinite sums of monomials) with coefficients in $K$ and 
indeterminates in $X$, with addition and multiplication defined in the obvious way.
If $K$ is $\omega$-continuous and $|X|=n$, then every formal
power series $f\in K[\![X]\!]$ induces a well-defined function $f: K^n\ra K$
which is $\omega$-continuous in each argument.
Further, if $K$ is $\omega$-continuous, then so is $K[\![X]\!]$ \cite{Kuich97}.
\edefn

A \emph{system of power series} with indeterminates $X_1\dots,X_n$  is a 
sequence $G=(g_1\dots g_n)$ with $g_i\in K[\![X]\!]$ for each $i$. It induces a
function $G:K^n\ra K^n$ that is monotone in each argument. By Kleene's Fixed-Point Theorem
$G$ has a least fixed point $\lfp G$ which coincides with the supremum of the Kleene approximants
$G^k$, defined by $G^0= 0$,  $G^{k+1}= G( G^k)$, i.e.  $\lfp G=\sup_{k\in\omega} G^k$.
We also refer to $\lfp G$ as the \emph{least fixed-point solution} of the equation system 
\[   X_1=g_1(X_1,\dots,x_n),\dots,   X_n=g_n(X_1,\dots,X_n), \]
in short, $X=G(X)$.

\topic{Dual-indeterminate power series}
Semirings $K[\![X]\!]$ of power series turn out to be appropriate as 
general provenance semirings for (not necessarily acyclic) reachability games, 
without any further structure on the terminal nodes, as well as
for purely positive fixed-point formalisms, without negation even on the atomic level.
However, as soon as we want to deal with fixed-point logics with (atomic) negation
we again need to take quotients with respect to the congruence generated
by an appropriate correspondence $X\leftrightarrow\nnX$
between positive and negative tokens (with the same conventions as in Definition~\ref{dual-indeterminate}).

\bdefn The semiring $K[\![X,\nnX]\!]$ is the quotient
of the semiring of power series  $K[\![X\cup\nnX]\!]$ by the
congruence generated by the equalities $p\cdot\nnp=0$ for
all $p\in X$. The 
congruence classes in $K[\![X,\nnX]\!]$ are in one-to-one correspondence
with the power series in $K[\![X\cup\nnX]\!]$
such that none of their monomials contain complementary tokens. 
We call these \emph{dual-indeterminate power series}.
\edefn

Again we have a universality property.
\begin{proposition}
Every function $f:X\cup\nnX\rightarrow K$ into an $\omega$-continuous
semiring $K$ with the property that 
$f(p)\cdot f(\nnp)=0$ for all $p\in X$ extends uniquely
to an $\omega$-continuous semiring homomorphism 
$h:\N[\![X,\nnX]\!]\rightarrow K$ that coincides with $f$ on $X\cup\nnX$.
\end{proposition}

\section{Provenance for reachability games with cycles}

We now extend our provenance approach to games that admit infinite plays. We assume that the game
graphs are finite, but no longer acyclic.
Given a valuation $f_\sigma:T\ra K$ in a semiring $K$ for the terminal nodes of a game graph $\Gg$, the rules defining 
valuations  for the other nodes have now to be read as an equation system $(G_\sigma)$ in indeterminates $X_v$ (for $v\in V$):
\begin{align*}
X_v &= f_\sigma(v) \quad\text{ for }v\in T\\ 
 (G_\sigma)\qquad\qquad X_v &= \sum_{w\in vE}  h_\sigma(vw)\cdot X_w \quad\text{ if }v\in V_\sigma \\
X_v &= \prod_{w\in vE}  h_\sigma(vw)\cdot X_w \quad\text{ if }v\in V_{1-\sigma}
\end{align*}

If we assume that the underlying semiring $K$ is $\omega$-continuous, then such a system $(G_\sigma)$ always has a least 
fixed-point solution $\lfp G_\sigma$, which can be computed as the limit of its Kleene approximants $G^n:V\ra K$, for $n\in\omega$.
These Kleene approximants can be seen as valuations in the unravellings $\Gg^n$ of the game $\Gg$ up to $n$ moves,
defined as follows.

Recall that, for every game graph  $\Gg=(V,V_0,V_1,T, E)$, and every initial  position $v_0\in V$,
we have  the \emph{tree unraveling} $\Tt(\Gg,v_0)=(V^\#,V_0^\#,V_1^\#,T^\#, E^\#)$ 
consisting of all finite paths from $v_0$, with the canonical projection 
$\rho:\Tt(\Gg,v_0)\rightarrow \Gg$ that maps every path $\pi v$ to its end point $v$.    

\bdefn 
Given $\Gg$ with basic valuations $f_\sigma:T\ra K$ and $h_\sigma: E\ra K\setminus\{0\}$ of the terminal positions and moves, 
the \emph{truncation} $\Gg^n=(V^{(n)},V_0^{(n)},V_1^{(n)},T^{(n)}, E^{(n)})$, for $n>0$, is the restriction 
of the union of the trees $\Tt(\Gg,v)$ (with $v\in V$) to paths of less than $n$ moves, and $\rho^n:\Gg^n\ra \Gg$
is the restriction of the canonical  homomorphism $\rho$ to $\Gg^n$.
Notice that the truncation induces new terminal nodes:  
\[ T^{(n)}:=\{\pi v\in V^{(n)}: v\in T\}\cup\{ \pi v  \in V^{(n)}: |\pi|=n-1, v\in V\setminus T\}.\]
In $\Gg^n$, we define the basic valuation of the moves, $h_\sigma^n:   E^{(n)}\ra K\setminus\{0\}$, in the obvious way,  by $h_\sigma^n(e):=h_\sigma(\rho^n(e))$. For the valuation of the terminal nodes $\pi v\in T^{(n)}$, we put
 $f^n_\sigma(\pi v)=f_\sigma(v)$ if $v\in T$, and  $f^n_\sigma(\pi v)=0$ otherwise, i.e. if $\pi v$ is an initial segment of a
 play in $\Gg$, with $n-1$ moves, that has not reached a terminal position in $T$.
\edefn

The games $\Gg^n$ are finite acyclic games, and the basic valuations extend to valuations $f^n_\sigma: V^{(n)}\ra K$
for all nodes of $\Gg^n$. By induction, it readily follows that, for all nodes $v$ of $\Gg$, the Kleene approximants $G^n$
of $(G_\sigma)$ coincide with these valuations.

\begin{lemma} \label{L1}  For all $n$ and all positions $v$ of $G$, we have that $G^n(v)=f^n_\sigma(v)$.
\end{lemma}

We denote the strategy space of Player~$\sigma$ from $v$ in $\Gg^n$ by $\Strat_\sigma^{(n)}(v)$.
Since the games $\Gg^n$ are acyclic, Theorem~\ref{thm:prov-games} applies.

\begin{lemma}  \label{L2} For every $n$, and every position $v$, $f^n_\sigma(v)= \sum_{\Tt\in\Strat_\sigma^{(n)}(v)} F(\Tt)$.
\end{lemma}

\medskip\noindent{\bf Valuations of plays and strategies in games with cycles. }
To generalize  Theorem~\ref{thm:prov-games} to reachability games with cycles, we first need to
extend the valuations of plays and strategies to such games.
As in Sect.~\ref{sect:wfgames} 
a \emph{finite} play $x=v_0v_1\dots v_m$ in $\Gg$ from $v_0$ to a terminal node $v_m\in T$
gets the valuation $f_\sigma(x)=h_\sigma(v_0v_1)\cdots h_\sigma(v_{m-1}v_m)\cdot f_\sigma(v_m)$.
The provenance value of an infinite play is defined to be 0. 
For a strategy $\Ss\in \Strat_\sigma(v)$, we put $F(\Ss):=0$
if $\Ss$ admits any infinite play. Hence
a strategy $\Ss$ can have a non-zero provenance value only when
it admits just finite plays. By K\"onig's Lemma, it then admits only a finite number
of plays, and putting, as in Sect.~\ref{sect:wfgames},
\[     F(\Ss):= \prod_{e\in E} h_\sigma(e)^{\#_\Ss(e)} \cdot \prod_{v\in T} f_\sigma(v)^{\#_\Ss(v)}\]
is well-defined for such strategies, as the values $\#_\Ss(e)$ and $ \#_\Ss(v)$ are finite,
for all $e\in E$ and $v\in T$.
Although the number of different strategies $\Ss\in\Strat_\sigma(v)$
may well be infinite, Theorem~\ref{thm:prov-games} generalizes to reachability games with cycles,
with a proof based on Kleene's fixed-point theorem, and the unravellings of $\Gg$ to finite acyclic games $\Gg^n$.

\medskip Notice that every strategy $\Ss\in\Strat_\sigma(v)$ for the original game $\Gg$ induces, in every game $\Gg^n$, 
a strategy $\Ss^{(n)}\in \Strat^{(n)}_\sigma(v)$ for the game $\Gg_n$.

\begin{lemma} \label{L3} For every strategy $\Ss\in\Strat_\sigma(v)$ in $\Gg$ with $F(S)\neq 0$ there exists some $n_\Ss<\omega$
such that
\begin{itemize}
\item $\Ss=\Ss^{(n)}$ for all $n\geq n_\Ss$,
\item $F(\Ss^{(m)})=0$ for all $m<n_\Ss$.
\end{itemize}
\end{lemma}

\lueg  This readily follows from the fact that a strategy $\Ss$ with $F(\Ss)\neq 0$ admits only a finite number of plays, all of which are finite.
Let $n_\Ss$ be the maximal length of these plays. Then, for $n\geq n_\Ss$, all plays in $\Ss$ are already contained in $\Ss^{(n)}$.
For any $m<n_s$, the induced strategy admits an unfinished play, hence $F(\Ss^m)=0$.
\hx

Every strategy $\Tt\in\Strat_\sigma^{(n)}(v)$ can be obtained as the induced strategy of some $\Ss\in\Strat_\sigma(v)$,
such that $\Tt=\Ss^{(n)}$.  In general $\Ss$ is not uniquely determined by $\Tt$ and $n$. Nevertheless, we have the following.

\begin{lemma} \label{L4} For every position $v$ of $\Gg$ and every $n<\omega$, we have that in $\Gg^n$,
\[  \sum_{\Ss\in\Strat_\sigma(v)}   F(\Ss^{(n)})  = \sum_{\Tt\in\Strat^{(n)}_\sigma(v)}   F(\Tt).\]
\end{lemma}

\lueg  If we have two strategies $\Ss_1 \neq \Ss_2$ in $\Strat_\sigma(v)$ with $\Tt=\Ss_1^{(n)}=\Ss_2^{(n)}$, then
$\Tt$ must contain an unfinished play (otherwise $\Tt=\Ss_1=\Ss_2$), which implies that $F(\Tt)=0$.
Thus, although the strategy spaces $\Strat_\sigma(v)$ are in general infinite, whereas $\Strat^{(n)}_\sigma(v)$
is finite for each fixed $n$, those strategies that provide non-zero values to the sums are in one-to-one correspondence,
and the two sums have the same value.  \hx

Putting these observations together, we obtain the desired generalization of Theorem~\ref{thm:prov-games}.

\begin{theorem}\label{thm:reachgames} For every game graph $\Gg$ with basic valuations $f_\sigma$ and $h_\sigma$ of the
terminal positions and moves in an $\omega$-continuous semiring $K$, we have that, for every position $v$
\[    f_\sigma(v):=(\lfp G_\sigma)(v)\ =\ \sum_{\Ss\in\Strat_\sigma(v)} F(\Ss).\]
In the cases where $h(e)=1$ for all $e$, or where 
$K$ is multiplicatively idempotent, we further have that
\[  f_\sigma(v)=\sum_{\Ss\in\Strat_\sigma(v)} \prod_ {x\in\Plays(\Ss)}  f_\sigma(x) .\]
\etheo

\lueg  By the lemmata above, we have that, for every $n<\omega$,
\[   G^n(v)=f^n_\sigma(v)= \sum_{\Tt\in\Strat_\sigma^{(n)}(v)} F(\Tt) = \sum_{\Ss\in\Strat_\sigma(v)}   F(\Ss^{(n)}).\]
Since, for every strategy $\Ss\in\Strat_\sigma(v)$ we have that $F(\Ss)=F(\Ss^n)$ for sufficiently large $n$,
the result follows by taking suprema.
\hx

For the case of game valuations $f_\sigma:V\ra\N[\![T]\!]$, given by the basic
valuations $f_\sigma(t)=t$ for terminal positions $t\in T$ and $h_\sigma(vw)=1$
for all moves $(v,w)\in E$, we again get precise information about the number 
of strategies that a player has for a specific outcome.
Indeed, $f_\sigma(v)$ is a (possibly) infinite sum of monomials
 $m\cdot t_1^{j_1}\dots t_k^{j_k}$

\begin{corollary} Let $f_\sigma:V\ra\N[\![T]\!]$ be the valuation of Player~$\sigma$ for
the game $\Gg$ in $\N[\![T]\!]$.
For every monomial  $m\cdot t_1^{j_1}\dots t_k^{j_k}$  in $f_\sigma(v)$ (with $m\in\N$ and $j_i>0$)  
Player~$\sigma$ has precisely $m$ strategies $\Ss$ from $v$ with the property that the
set of possible outcomes for $\Ss$ is precisely $\{t_1,\dots,t_k\}$, and 
precisely $j_i$ plays that are consistent with $\Ss$ have the outcome $t_i$.
\end{corollary}

Let $\Gg=(V,V_0,V_1,T,E)$ be a game with reachability objectives $T_0,T_1$ for the two players,
such that $T_0\cap T_1=\emptyset$. Let $W_0, W_1\subseteq V$ be the winning regions
for the two players, i.e., $W_\sigma$ is the set of those positions $v\in V$ such that
Player~$\sigma$ has a strategy from $v$ to force the play to $T_\sigma$.
Note that $V$ is the disjoint union of the $W_0$, $W_1$ and $U$, the set of those positions
from which none of the two players has a winning strategy. By Zermelo's Theorem both
players have strategies to guarantee that each play from $U$ will be at least a draw.

\begin{corollary}\label{reachability-valuation} Let $f_\sigma:T\ra K$ be a valuation of the terminal positions of 
$\Gg$ in an $\omega$-continuous semiring, with $f_\sigma(t)\neq 0$ if, and only if, $t\in T_\sigma$.
The least fixed point solution of the equation system $F_\sigma$ extends this to a
valuation $f_\sigma:V\ra K$, with $f_\sigma(v)\neq 0$ if, and only if, $v\in W_\sigma$.
\end{corollary}

Notice that weakly contradictory valuations $f_0$ an $f_1$ on the terminal positions
extend to weakly contradictory valuations on all positions. 
However, even valuations into $\omega$-continuous semirings that
are strongly contradictory on the terminal positions, are in general only weakly contradictory
on the set of all positions, unless $W_0\cup W_1=V$, since $f_0(U)=f_1(U)=0$.
\ex\label{ex:reachability-game}
We illustrate our findings by the following very simple example of a game where Player~0 moves from $v$,
Player~1 moves from $w$, and $s$ and $t$ are terminal nodes.

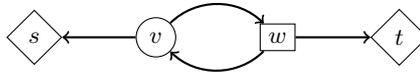
\begin{center}
\begin{tikzpicture}[baseline=-4pt,scale=.8]
        \tikzstyle{every node}=[shape=circle, draw=black]
        \node[shape = diamond] (0) at (0,0) {$s$};
        \node (1) at (2,0) {$v$};
        \node[shape = rectangle] (2) at (4,0) {$w$};
        \node[shape = diamond] (3) at (6,0) {$t$};
        \draw[->,thick,bend left=45] (1) edge (2) (2) edge (1);
        \draw[->,thick] (1) edge (0) (2) edge (3);
      \end{tikzpicture}
\end{center}

The corresponding equation system for Player~0 has the equations $X_v=s+X_w$ and $X_w=t\cdot X_v$.
In $\N^\infty[\![s,t]\!]$ the least fixed-point solution is $f(v)=s\cdot(1+t+t^2\dots)$ and
$f(w)=s\cdot(t+t^2+\dots)$. If we evaluate it for the reachability objectives $\{s\}$ and $\{t\}$, respectively,
we obtain $f(v)[0,t]=f(w)[0,t]=0$ which illustrates that neither from $v$ nor from $w$,
Player~0 has a strategy to reach $t$. On the other side, $f(v)[s,0]=s$ and $f(w)[s,0]=0$
which is consistent with the fact that Player~0 has a strategy to reach $s$ from $v$ but
not from $w$.

But the formal power series $f(v)$ and $f(w)$ reveal more information than that.
For instance, the fact that $f(v)$ contains, for every $n$, the monomial $s\cdot t^n$ 
implies that Player~0 has precisely one strategy $S$ from $v$ that admits
precisely $n+1$ consistent plays, one of which has outcome $s$ and the other $n$
have outcome $t$; this is the strategy where Player~0 moves from $v$ to $w$ the first
$n$ times, and then to $s$.
Notice that Player~0 also has one further strategy, namely the (positional) strategy to move
always to $w$. However, this strategy does not guarantee that the play terminates and therefore
has value 0, so it is not visible in the provenance values $f(v)$ and $f(w)$.
\eex

\section{Provenance analysis for positive LFP}

Least fixed-point logic, denoted LFP,   
extends first order logic by 
least and greatest fixed points of definable 
monotone operators on relations:
If  $\psi(R,\bar x)$ is a formula 
of vocabulary $\tau\cup\{R\}$, in which the  relational variable $R$ occurs only positively, and 
if $\bar x$ is a tuple of variables 
such that the length of $\bar x$ matches the arity of $R$, then 
$[\lfp R\bar x\st \psi](\bar x)$ and $[\gfp R\bar x \st\psi](\bar x)$
are also formulae (of vocabulary $\tau$).
The semantics of these formulae is that $\bar x$ is contained in the least
(respectively the greatest) fixed point of the update operator
$F_\psi:R\mapsto\{\bar a:  \psi(R,\bar a)\}$. Due to the positivity of $R$ in $\psi$,
any such operator $F_\psi$ is monotone and therefore has, by the Knaster-Tarski-Theorem, a least 
fixed point $\lfp(F_\psi)$ and a greatest fixed point $\gfp(F_\psi)$. 
See e.g. \cite{Graedel+07} for background on $\LFP$.
 
Note that in formulae
$[\lfp R\bar x \st \psi](\bar x)$ one may allow $\psi$ to
have other free variables besides $\bar x$; these are called parameters
of the fixed-point formula. However, at the expense of increasing the
arity of the fixed-point predicates and the number of variables one can always eliminate parameters.
For the construction of model-checking games and also for provenance analysis it is convenient to assume 
that formulae are parameter-free.
The duality between least and greatest fixed point implies
that for any  $\psi$,
\[ [\gfp R\bar x\st \psi](\bar x)\equiv \neg [\lfp R\bar x\st \neg\psi[R/\neg R]](\bar x).\]
Using this duality together with de Morgan's laws, every LFP-formula can
be brought into \emph{negation normal form}, where negation applies to
atoms only.

\medskip\noindent{\bf The fragment of positive least fixed points. }
We denote by $\posLFP$ the fragment of LFP consisting of formulae in negation normal form
such that all its fixed-point operators are least fixed-points.
It is known that, on finite structures (but not in general), $\posLFP$ has the same
expressive power as full LFP, and thus captures all polynomial-time computable
properties of ordered finite structures \cite{Graedel+07} . 

An advantage of dealing with $\posLFP$, rather than full LFP, is that it admits much simpler model checking
games. Indeed the appropriate games for LFP are \emph{parity games},
whereas for $\posLFP$, reachability games are sufficient. This can be exploited 
to define provenance interpretations for fixed-point formulae, along
the lines described in the previous section.

Definition~\ref{FOgame} of model checking games $\Gg(A,\psi)$ 
for $\psi\in\FO(\tau)$ extends to 
formulae $\psi(\bar x)\in\posLFP(\tau)$ as follows:
For every subformula of $\psi$ of form
$\theta:=[\lfp R\bar x \st \phi(R,\bar x)](\bar x)$
we add moves from positions $\theta(\bar a)$ to $\phi(\bar a)$, and from
positions $R\bar a$  to $\phi(\bar a)$ for  every tuple $\bar a$. 
Since these moves are unique it makes no difference to which of the two players
we assign the positions $\theta(\bar a)$ and $R\bar a$.
The resulting game graphs $\Gg(A,\psi)$ may contain cycles,
but the set $T$ of terminal nodes is again a subset of $\Lit_A(\tau)$. 

A $K$-interpretation $\pi:\Lit_A(\tau)\ra K$ into
an $\omega$-continuous semiring thus provides a valuation of the 
terminal positions of the game graph $\Gg(A,\psi)$ for any
$\psi\in\posLFP(\tau)$. 
By Theorem~\ref{reachability-valuation} this extends to a valuation $f_0: V\ra K$ on
the set $V$ of all positions $\phi(\bar a)$ of $\Gg(A,\psi)$, including position $\psi$ itself.

\bdefn  For any instantiated subformula $\phi$ of a sentence $\psi\in\posLFP$, we define the provenance value $\pi\ext{\phi}$ by its game valuation:
$\pi\ext{\phi}:= f_0(\phi)$. \edefn

In particular, if $\pi$ is model-defining, then $f_0$ provides truth values for all
fully instantiated subformula $\phi$ of $\psi$ on the structure $\AA_\pi$
that $\pi$ describes. Indeed $\AA_\pi\models \phi$
if, and only if, $\pi\ext{\phi}\neq 0$, and in that case the value $\pi\ext{\phi}$ gives
us additional information, how and why $\phi$ holds in $\AA$, for instance by 
information on the winning strategies that Verifier has available for establishing the truth of $\phi$ in
$\AA_\pi$. 
However, contrary to the case of
first-order logic, in the case where $\AA_\pi\not\models \phi$, and hence
$\pi\ext{\phi}=0$, we do not get additional information on the reasons why
$\phi$ is false. The possibility to move to $\neg\phi$ (or more precisely, its negation normal form)
and to do the provenance analysis for that formula, does not exist here since $\neg\phi$ is
not a formula of posLFP.  In fact, the model checking-game for $\neg\phi$ is not
a reachability game, but a safety game. To deal with safety games and greatest fixed points
we shall have to impose additional restrictions on the underlying semirings.
We shall discuss this below.

\medskip
One can define provenance values for posLFP-sentences also directly
by a  fixed-point interpretation in $\omega$-commutative
semirings.
The goal is to extend, by induction over the syntax, a $K$-interpretation $\pi:\Lit_A(\tau)\ra K$ to
valuations $\pi\ext{\psi}\in K$ for all sentences $\psi\in\posLFP(\tau\cup A)$.
The rules for first-order operations are defined already, so we just have to consider
sentences of form $\psi(\bar a)=[\lfp R\bar x. \phi(R,\bar x)](\bar a)$,
with $\phi \in \posLFP(\tau\cup\{R\})$.
If $R$ has arity $m$, then its $K$-interpretations of $A$ are functions $g:A^m\ra K$.
These functions are ordered, by $g\leq g'$ if, and only if, $g(\bar a)\leq g'(\bar a)$
for all $\bar a\in A^m$.
Given a $K$-interpretation $\pi:\Lit_A(\tau)\ra K$, we denote by $\pi[R\mapsto g]$
the $K$-interpretation of $\Lit_A(\tau)\cup \Atoms_A(\{R\})$  obtained from $\pi$ by
adding values $g(\bar c)$ for the atoms $R\bar c$. (Notice that $R$ appears only positively in $\phi$,
so negated atoms are not needed).

The formula $\phi(R,\bar x)$ now defines, together with $\pi$, a monotone update
operator $F_\pi^\phi$ on functions $g: A^m\ra K$. More precisely, it maps
$g$ to 
\[   F_\pi^\phi(g): \bar a\mapsto   \pi[R\mapsto g]\ext{\phi(R,\bar a)}. \]
By Kleene's Fixed-Point Theorem, the operator $F_\pi^\phi$ has a least fixed point
$\lfp(F_\pi^\phi)$ which coincides with the limit of the sequence $(g^n)_{n<\omega}$
with $g^0:=0$ and $g^{n+1}:= F^\phi_\pi(g^n)$, and which we may define as
the provenance value of $[\lfp R\bar x. \phi(R,\bar x)](\bar a)$.
The two definitions coincide.

\begin{proposition} For every formula $ [\lfp R\bar x. \phi(R,\bar x)]\in\posLFP$
and every $K$-interpretation $\pi:\Lit_A(\tau)\ra K$ into an $\omega$-continuous semiring, 
$\pi\ext{\, [\lfp R\bar x. \phi(R,\bar x)](\bar a)} = \lfp(F_\pi^\phi)(\bar a).$
\end{proposition}
   
The proof is a rather straightforward adaptation of the correctness proof for model checking games
for LFP, see e.g. \cite[Chapter 3.3]{Graedel+07}.

\section{Beyond reachability: safety games and greatest fixed points} 

While the restriction of LFP to its positive fragment comes with no loss of expressive
power (on finite structures) and while $\posLFP$ is sufficiently powerful to capture a number of
interesting and relevant other fixed-point formalisms in computer science,
it is nevertheless not really satisfactory. One reason is that the transformation
from a fixed-point formula with non-atomic negation into one in $\posLFP$
is (contrary to transformations into negation normal form) 
not a simple syntactic translation. It goes through the Stage Comparison Theorem and
can make a formula much longer and more complicated. Further, such transformations
are  not available for important fixed-point formalism such as the modal $\mu$-calculus,
stratified Datalog, transitive closure logics, and even simple temporal languages
such as CTL. On the game-theoretic side, reachability games are just the simplest kind
of games on graphs, and in many applications players have different and more ambitious goals  
such as safety, B\"uchi, parity or Muller, objectives.  It is thus an important and interesting challenge to lay the foundations
of a provenance analysis for full LFP and infinite games with more general objectives, and to apply this approach
to the numerous other fixed-point formalisms, in particular in databases and verification.  

We defer a detailed treatment of this to forthcoming work. Here we discuss some of the mathematical concepts and
challenges that arise in this project, and apply them to the provenance of \emph{safety games}.
Recall that the computation of winning positions for safety objectives is a simple, but also
in some sense universal, application of greatest fixed points.   

The first observation is that we need to impose additional requirements on the semirings that we consider.
While $\omega$-continuous semirings are appropriate for a provenance analysis of
least fixed points and reachability objectives, they are not always adequate for greatest
fixed points. The property of $\omega$-continuity is not sufficient to guarantee the
existence of greatest fixed points, and in cases where they exist they  do not necessarily provide the
information that we are interested in. 

\vspace*{1cm}

\ex\label{ex:safety-game}
We consider the game graph
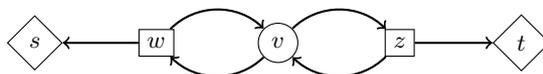
\begin{center}
\begin{tikzpicture}[baseline=-4pt,scale=.8]
        \tikzstyle{every node}=[shape=circle, draw=black]
        \node[shape = diamond] (0) at (0,0) {$s$};
        \node[shape = rectangle] (1) at (2,0) {$w$};
        \node (2) at (4,0) {$v$};
        \node[shape = rectangle] (3) at (6,0) {$z$};
        \node[shape = diamond] (4) at (8,0) {$t$};
        \draw[->,thick,bend left=45] (1) edge (2) (2) edge (1) (2) edge (3) (3) edge (2);
        \draw[->,thick] (1) edge (0) (3) edge (4);
      \end{tikzpicture}
\end{center}
with associated equation system for Player~0 consisting of $X_v=X_w+X_z$, $X_w=f(s)\cdot X_v$, and $X_z=f(t)\cdot X_v$.
The least fixed-point solution (in whatever semiring) has values $f(v)=f(w)=f(z)=0$ which reflects the fact that Player~0
has no strategy to guarantee a finite play.
It is not difficult to see that in $\N^\infty[\![s,t]\!]$ this in fact the unique fixed point, hence in particular the greatest one,
which however gives us no information about safety strategies. In $\N^\infty$ instead, under a valuation of the terminal
nodes with $f(s)=a\neq 0$ and $f(t)=0$, we get the greatest fixed point $f(v)=f(w)=\infty$ and $f(z)=0$. In particular,
greatest fixed-points do not specialise correctly from $\N^\infty[\![s,t]\!]$ to $\N^\infty$.

We shall see below that get interesting information on safety strategies by provenance values in the absorptive semiring
$\S^\infty[s,t]$.
\eex

To make sure that also greatest fixed points of polynomial equation systems exist, we shall require that our
semirings are not just $\omega$-continuous, but also also \emph{$\omega$-co-continuous},
i.e. that every descending $\omega$-chain $(a_i)_{i\in\omega}$, with $a_{i+1}\leq a_{i}$ for all $i\in\omega$,
has an infimum $\inf_{i\in\omega} a_i$  in $K$, which is compatible with the
semiring operations in the sense that, for every $c\in K$,
\[  c + \inf_{i\in\omega} a_i = \inf_{i\in\omega} (c +a_i) \text { and } c\cdot \inf_{i\in\omega} a_i=\inf_{i\in\omega} (c\cdot a_i).\]
We call such semirings \emph{fully $\omega$-continuous}.
Our most important example of such a semiring is $\Sinf[X]$, the semiring of generalized
absorptive polynomials, that we are going to discuss next.

\section{Absorptive semirings and generalized absorptive polynomials}

Recall that a semiring $K$ is \emph{absorptive} if $a+ab=a$ for all $a,b\in K$
which is equivalent to $1+a=1$ for all $a\in K$.
Examples include the Viterbi semiring, 
the tropical semiring, min-max semirings,
further the semiring $\S[X]$ of absorptive polynomials over $X$.
Absorptive semirings are +-idempotent and naturally ordered, 1 is the top element,
and multiplication decreases elements: $ab\leq b$.
In particular, the powers of an element form a descending
$\omega$-chain $1 \geq a\geq a^2 \geq \cdots$.
If this chain has an infimum then we denote it by $a^\infty$.

In the semiring $\S[X]$, the infima of descending $\omega$-chains $(x^n)_{n<\omega}$ are always 0 and thus not 
very informative. 
We therefore complete $\S[X]$ to the semiring $\Sinf[X]$ by admitting exponents in $\N^\infty$.

\bdefn Let $X$ be a \emph{finite} set of provenance tokens. 
A \emph{monomial} over $X$ with exponents from
$\Ninf$ is a function $m:X\rightarrow\Ninf$. 
Informally, we write $m$ as $x_1^{m(x_1)}\cdots x_n^{m(x_n)}$. 
Monomial multiplication adds the exponents. 
Observe also that $x^\infty\cdot x^n =x^\infty$.
For any two monomials, $m_1,m_2$
we say that that $m_2$ \emph{absorbs} $m_1$ if $m_2$ has smaller exponents
than $m_1$. Formally, $m_1\preceq m_2$ if, and only if, $m_1(x)\geq m_2(x)$ for all $x\in X$.
Since monomials are functions, this is the pointwise
partial order given by the order on $\Ninf$.
\edefn

Because $\Ninf$ is a lattice (with top and bottom) the monomials
also inherit a lattice structure.
The set of all monomials is, of course, infinite. However,
it has some crucial finiteness properties.

\begin{proposition}
Every ascending chain and every antichain of monomials is finite.
\end{proposition}

\begin{proof} Clearly $(\Ninf,\leq)$ is a well-order. For any finite set $X$,
the set of monomials $m:X\rightarrow \Ninf$ with the \emph{reverse order}
of the absorption order is isomorphic to $(\Ninf)^k$ with $k=|X|$ and
with the component-wise order inherited from $(\Ninf,\leq)$. This is a well-quasi-order
and therefore has no infinite descending chains and no infinite antichains.
This implies that in the set of monomials over $X$ with the absorption order,
all ascending chains and all antichains are finite.  
\end{proof}

\bdefn
We define $\Sinf[X]$ as the set of antichains of monomials
with indeterminates from $X$ and exponents in $\Ninf$.
Writing an antichain as a (formal) sum of its monomials we identify
it with a polynomial with coefficients 0 or 1, and call these
\emph{generalized absorptive polynomials}. We define polynomial addition and multiplication as usual, 
except that for coefficients 1+1=1,  and that we keep only the maximal monomials in the result.
The empty antichain corresponds to the 0 polynomial. The 1 polynomial
consists of just the monomial in which every indeterminate has exponent 0. 
\edefn

\begin{proposition}
$(\Sinf[X],+,\cdot,0,1)$ is an absorptive commutative semiring.
Further it is a complete lattice wrt. to the natural order, which is 
fully $\omega$-continuous and moreover
completely distributive.
\end{proposition}

As a consequence, we can compute not only least fixed point solutions
for systems of polynomial equations but also greatest fixed points.
In contrast to other semirings with such properties, such as for instance the
Viterbi semiring, $\Sinf[X]$ has one further crucial property.
It is \emph{chain-positive} which means that the infimum of 
every chain of non-zero elements is also non-zero.

As in other semirings of polynomials and power series we can also here 
take pairs of positive and negative indeterminates, with
a correspondence $X\iff \nnX$ and build
the quotient with respect to the congruence generated by 
the equation $x\cdot\nnx =0$.
We thus obtain a new semiring $\Sinf[X,\nnX]$ which provides  
a natural framework for a provenance analysis for
full LFP and other fixed point calculi. We shall develop this in
forthcoming work.

Here we use the semiring $\Sinf[T]$ to describe a provenance analysis for
safety games where $T$ is the set of terminal positions of the given game graph. 

\section{Absorption among strategies}\label{sect:strategy-absorption}

\bdefn
Let  $\Gg=(V,V_0,V_1,T,E)$ be a finite game graph,
and $v\in V$.
For two strategies $\Ss,\Ss'\in\Strat_\sigma(v)$, we say that 
$\Ss$ \emph{absorbs} $\Ss'$ (in symbols $\Ss\succeq_a \Ss'$) 
if
\begin{itemize}
\item for all $t\in T$, $\Ss$ admits at most as many plays
with outcome $t$ as $\Ss'$ does, and
\item if $\Ss$ admits an infinite play, then so does $\Ss'$.
\end{itemize}
We call $\Ss$ \emph{absorption-dominant} 
if it is maximal with respect to  $\succeq_a$. 
\edefn

Absorption-dominant strategies are interesting both for games in general
and for logic because they can win ``with minimal effort''.
As a simple example, consider a model checking game for a formula $\phi\lor(\phi\land\psi)$.
The Verifier can either establish $\phi$ or $\phi\land\psi$, but any strategy
that establishes the truth of $\phi\land\psi$ will have more plays and 
more outcomes than one that proves just $\phi$, and will thus be absorbed by it. 
The  absorption-dominant strategies for  $\phi\lor(\phi\land\psi)$
are thus precisely the  absorption-dominant strategies for  $\phi$.

Notice however that, despite this minimality, absorption dominant strategies
need not be positional, not even in acyclic games.

\ex\label{absdom-vs-pos}
Consider the game
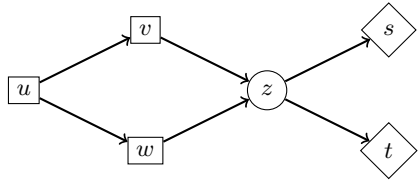
\begin{center}
\begin{tikzpicture}[baseline=-4pt,scale=.8]
        \tikzstyle{every node}=[shape=circle, draw=black]
        \node[shape = rectangle] (0) at (0,0) {$u$};
        \node[shape = rectangle] (1) at (2,1) {$v$};
        \node[shape = rectangle] (2) at (2,-1) {$w$};
        \node (5) at (4,0) {$z$};
        \node[shape = diamond] (6) at (6,1) {$s$};
        \node[shape = diamond] (7) at (6,-1) {$t$};
        \draw[->,thick] (0) edge (1) 
                                 (0) edge (2) 
                                 (1) edge (5) 
                                 (2) edge (5) 
                                 (5) edge (6) 
                                 (5) edge (7);
      \end{tikzpicture}
\end{center}
There are four strategies in $\Strat_0(u)$ with provenance values $s^2$, $st$, $st$, and $t^2$.
The positional ones are those with values $s^2$ and $t^2$, but all four strategies are
absorption-dominant.
\eex

However, absorption-dominant strategies are \emph{weakly positional} in the 
sense that if a node is reached several times during the same play,
then, without loss of strategic power, the player can  
always make the same choice at that node.
Absorption  among strategies makes sense for both acyclic and
cyclic games. In acyclic games, absorption-dominant strategies are described 
by provenance polynomials in $\S[T]$ (with only finite exponents).
But they are even more interesting for the analysis of
reachability \emph{and safety} games that admit infinite plays. 
The fundamental difference between valuations for 
reachability and safety strategies concerns the valuations of infinite
plays. If, as we assume here, reachability and safety goals are defined for terminal nodes,
then an infinite play is losing for every reachability objective but winning for
every safety objective. As a consequence,
the strategies $\Ss\in\Strat_\sigma(v)$ that enforce all
plays to be non-terminating absorb all other strategies
in $\Strat_\sigma(v)$ that admit at least one infinite play. 

We thus extend the valuations of plays in a game $\Gg$ (with finite game graph that may contain cycles) 
to two different valuation function $f_\sigma^\mu$ and  $f_\sigma^\nu$.
For simplicity, we assume trivial valuations on the edges, so for
a finite play $x$ ending in $t$, we just put
$f_\sigma^\mu(x)=f_\sigma^\nu(x)=f_\sigma(t)$  
but if $x$ is an infinite play, we put $f_\sigma^\mu(x)=0$ and $f_\sigma^\nu(x)=1$.

A strategy $\Ss\in\Strat_\sigma(v)$ may well admit an infinite set of plays.
Taking the semiring  $\Sinf[T]$ with the basic valuation $f_\sigma(t):=t$
for the terminal nodes, strategies are described by monomials  (or 0), and we put 
\[   F^\mu(\Ss):=\prod_{x\in\Plays(\Ss)}  f_\sigma^\mu(x) \qquad\text{ and }\qquad   
 F^\nu(\Ss):=\prod_{x\in\Plays(\Ss)}  f_\sigma^\nu(x) =  \prod_{t\in T} t^{\#_\Ss(t)}.\]
 We extend the absorption order $\succeq$ on monomials by $m\succeq 0$ for all $m$.

 \begin{lemma}  Let $\Gg$ be any finite game graph, with valuations 
 of strategies in $\Sinf[T]$ induced by $f_\sigma(t)=t$ for all $t\in T$. 
 For all strategies $\Ss, \Ss'\in\Strat_\sigma(v)$ we have,
 \begin{itemize}
 \item  $0\neq F^\nu(\Ss)\succeq F^\mu(\Ss)\neq 1$.
 \item  $F^\mu(\Ss)= 0$  if, and only if, $\Ss$ admits an infinite play. Otherwise $F^\mu(\Ss)=F^\nu(\Ss)$. 
 \item  $F^\nu(\Ss)= 1$  if, and only if, $\Ss$ admits only infinite plays.
 \item $\Ss$ absorbs $\Ss'$ if, and only if, both $F^\nu(\Ss)\succeq F^\nu(\Ss')$ and
 $F^\mu(\Ss)\succeq F^\mu(\Ss')$.
 \end{itemize}
 \end{lemma}
   
 \lueg  Only the last item requires proof. Suppose that $\Ss$ absorbs $\Ss'$. If
 $\Ss$ admits only finite plays, then $F^\mu(\Ss)=F^\nu(\Ss)\succeq F^\nu(\Ss')\succeq F^\mu(\Ss')$.
 If $\Ss$ admits an infinite play, then so does $\Ss'$ and $F^\nu(\Ss)\succeq F^\nu(\Ss')\succeq F^\mu(\Ss')=F^\mu(\Ss)=0$.
 In both cases, $F^\nu(\Ss)\succeq F^\nu(\Ss')$ and
 $F^\mu(\Ss)\succeq F^\mu(\Ss')$.
 Conversely, assume that $\Ss$ does not absorb $\Ss'$. Then either there is a terminal $t$ such that $\Ss$ admits more plays
 with outcome $t$ than $\Ss'$ does, or $\Ss$ admits an infinite play, but $\Ss'$ does not. In the first case, $F^\nu(\Ss)\not\succeq F^\nu(\Ss')$ 
 and in the second case $0=F^\mu(\Ss)\not\succeq F^\mu(\Ss')\neq 0$.
  \hx

\ex
We return to the game described in Example~\ref{ex:reachability-game}
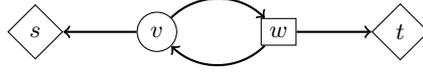
\begin{center}
\begin{tikzpicture}[baseline=-4pt,scale=.8]
        \tikzstyle{every node}=[shape=circle, draw=black]
        \node[shape = diamond] (0) at (0,0) {$s$};
        \node (1) at (2,0) {$v$};
        \node[shape = rectangle] (2) at (4,0) {$w$};
        \node[shape = diamond] (3) at (6,0) {$t$};
        \draw[->,thick,bend left=45] (1) edge (2) (2) edge (1);
        \draw[->,thick] (1) edge (0) (2) edge (3);
      \end{tikzpicture}
\end{center}
with equation system $G_0$ consisting of $X_v=s+X_w$ and $X_w=t\cdot X_v$.
In $\N^\infty[\![s,t]\!]$ the least fixed-point solution is $f(v)=s\cdot(1+t+t^2\dots)$ and
$f(w)=s\cdot(t+t^2+\dots)$. In  $\Sinf[s,t]$ the least fixed-point 
solution $f^\mu=\lfp G_0$ has values  $f^\mu(v)=s$ and $f^\mu(w)=st$, which describes the possible outcomes
of the unique absorption-dominant dominant strategy that enforces finite plays.
The only other absorption-dominant strategy (moving from $v$ to $w$) has value 0
because it admits an infinite play.       

However, the greatest fixed-point solution $f^\nu=\gfp G_0$ of this equation system  in $\Sinf[s,t]$
has values  $f^\nu(v)=s+ t^\infty$ and $f^\nu(w)=st+t^\infty$. Here this second strategy has value
$t^\infty$ since it admits infinitely many plays ending in $t$ (and one infinite play
with value 1).
\eex


\begin{theorem}  Let $\Gg=(V,V_0,V_1,T,E)$  be a game graph and let $G_\sigma$ be the associated 
equation system for Player~$\sigma$. In the semiring $\Sinf[T]$ this system has least and greatest 
fixed point solutions $\lfp G_\sigma$ and $\gfp G_\sigma$ with
\[
(\lfp G_\sigma)(v):=\sum_{\Ss\in\Strat_\sigma(v)}  F^\mu(\Ss) \qquad\text{ and }\qquad 
(\gfp G_\sigma)(v):=\sum_{\Ss\in\Strat_\sigma(v)}  F^\nu(\Ss).
\]
The values of these sums do not change if we restrict them to the absorption-dominant strategies.
\end{theorem}

\lueg Since $\Sinf[T]$ is $\omega$-continuous,  the claim for $(\lfp G_\sigma)$ follows from
Theorem~\ref{thm:reachgames}. For the greatest fixed-point solution we use that 
$\Sinf[T]$ is also $\omega$-co-continuous and has the structure of a complete lattice.
Thus, $(\gfp G_\sigma)$ is the limit of the descending chain $(G^n)_{n<\omega}$ of approximants
starting with $G^0=1$, and $G^{n+1}$ is defined by applying the equation system 
to $G^n:V\ra\Sinf[T]$.
   
As in the proof of Theorem~\ref{thm:reachgames} we argue with the unfoldings $\Gg^n$ of $\Gg$
up to $n-1$ moves, but we now put $f^n_\sigma(\pi v)=1$ for the final node of an `unfinished'
play, i.e. with $|\pi|=n-1$ and $v\in V\setminus T$. The valuations $f^n_\sigma$ extend to
all nodes of the (acyclic) game $\Gg^n$, and again, coincide with the Kleene approximants $G^n$:
for every $n$ and every $v$ we have that $G^n(v)=f_\sigma^n(v)$.
The different valuation of the terminal nodes in $\Gg^n$ also has the effect 
that any $\Tt\in\Strat^{(n)}_\sigma(v)$ we have that $F(\Tt)=\prod_{t\in T} t^{\#_\Tt(t)}$,
which is a monomial with only finite exponents. Since $\Gg^n$ is acyclic
\[    f^n_\sigma(v)=\sum_ {\Tt\in\Strat^{(n)}_\sigma(v)} F(\Tt).\]

Every strategy $\Ss\in\Strat_\sigma(v)$ for the original game $\Gg$ induces for each game $\Gg^n$
a strategy $\Ss^{(n)})$. Conversely,  every strategy $\Tt\in\Strat^{(n)}_\sigma(v)$ is induced by at least one 
strategy $\Ss\in\Strat_\sigma(v)$. Since the semiring $\Sinf[T]$ is idempotent, we have that, for every $n<\omega$,
\[  \sum_{\Ss\in\Strat_\sigma(v)}   F(\Ss^{(n)})  = \sum_{\Tt\in\Strat^{(n)}_\sigma(v)}   F(\Tt).\] 

As graphs, the sequence $(\Ss^{(n)})_{n\in\omega}$ of induced strategies is increasing, i.e. $\Ss^{(1)}\subseteq S^{(2)} \subseteq \cdots$, 
but their values in $\Sinf[T]$ are decreasing, i.e.,   $F(\Ss^{(1)})\succeq F(S^{(2)}) \subseteq \cdots$.
Further, $F^\nu(\Ss) =  \prod_{t\in T} t^{\#_\Ss(t)}$  with exponents $\#_\Ss(t) \in\Ninf$ and
the corresponding exponents in the monomial $F(\Ss^{(n)}$ tell us how often
a terminal position $t\in T$ has been reached by $\Ss$ after $n-1$ moves. In particular,
\[  F^\nu(\Ss) = \lim_{n\ra\infty}  F(\Ss^{(n)}).\]
It is clear that these limits commute with summation over stratgies, so we have that
\[  \sum_{\Ss\in\Strat_\sigma(v)}  F^\nu(\Ss) = \lim_{n\ra\infty} \sum_{\Ss\in\Strat_\sigma(v)}  F(\Ss^{(n)}) =   
 \lim_{n\ra\infty}  \sum_{\Tt\in\Strat^{(n)}_\sigma(v)}   F(\Tt).\]
Putting everything together we get that
\[  (\gfp G_\sigma)(v) = \lim_{n\ra\infty} G^n((v) = \lim_{n\ra\infty} f^n_\sigma(v) = 
\lim_{n\ra\infty} \sum_ {\Tt\in\Strat^{(n)}} F(\Tt)  = \sum_{\Ss\in\Strat_\sigma(v)}  F^\nu(\Ss).\]
\hx

These least and greatest fixed points give precise descriptions of the absorption-dominant reachability and 
safety strategies of the players for each position of the game.

\ex
We return to the Example~\ref{ex:safety-game}: 
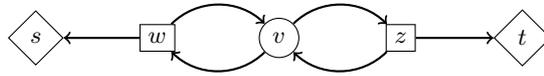
\begin{center}
\begin{tikzpicture}[baseline=-4pt,scale=.8]
        \tikzstyle{every node}=[shape=circle, draw=black]
        \node[shape = diamond] (0) at (0,0) {$s$};
        \node[shape = rectangle] (1) at (2,0) {$w$};
        \node (2) at (4,0) {$v$};
        \node[shape = rectangle] (3) at (6,0) {$z$};
        \node[shape = diamond] (4) at (8,0) {$t$};
        \draw[->,thick,bend left=45] (1) edge (2) (2) edge (1) (2) edge (3) (3) edge (2);
        \draw[->,thick] (1) edge (0) (3) edge (4);
      \end{tikzpicture}
\end{center}
Recall that the associated equation system for Player~0 has the equations
$X_v=X_w+X_z$, $X_w=f(s)\cdot X_v$, and $X_z=f(t)\cdot X_v$. 

The greatest fixed-point solution in $\S^\infty[s,t]$, computed by iterating
from the top element $f=1$ results in $f^\nu(v)=s^\infty  + t^\infty$, $f^\nu(w)=s^\infty+st^\infty$, and
$f^\nu(z)=s^\infty t + t^\infty$. Notice that indeed, $f^\nu(v)=f^\nu(w)+f^\nu(z)$ because $st^\infty$ is absorbed by
$t^\infty$, and $s^\infty t$ by $s^\infty$.  The greatest fixed point solution indicates that
Player~0 has two absorptive strategies (move always to $w$ or move always to $z$),
and gives, for each of the terminal nodes $s$ and $t$ the number of plays ending
in that node that the strategy admits. For instance, if the safety objective requires to
avoid $t$, then $v$ and $w$ the strategy moving to $w$ has infinitely many 
winning plays ending in $s$ (and one nonterminating play with value 1),
but since $f(z)[s,0]=0$, Player 0 has no safety strategy from $z$ that avoids $t$.
\eex

\section{Outlook}
In this paper we have extended the semiring framework for provenance analysis by new elements,
so that it can be applied to logics with negation, in particular first-order logic and
fixed-point logics, and to an analysis of games that provides detailed information about
the number and properties of the strategies of the players. 

Our treatment of negation is based on transformations to negation normal form and the use
of newly introduced semirings of dual-indeterminate polynomials and dual-indeterminate power series.
In particular, $\omega$-continuous semirings $\N^\infty[\![X,\nnX]$ of dual-indeterminate power series
povide an adequate general framework for logics with least fixed points, such as $\posLFP$ (and Datalog)
and the semiring of absorptive generalized dual-indeterminate polynomials $\Sinf[\![X,\nnX]$ 
permits an adequate treatment of greatest fixed points.  We have thus laid foundations for 
a provenance analysis of general fixed-point logics, and we are currently applying this also to
modal, temporal, and dynamic logics.

On the level of games, we have seen that provenance valuations in $\omega$-continuous and absorptive
semirings give us very detailed information about strategies for possibly infinite games 
with reachability and safety objectives. We are currently expanding this to games with more complicated
objectives, such as B\"{u}chi, Co-B\"{u}chi or parity games. Since these objectives do no longer depend on
terminal nodes but on the data occurring in infinite plays, a somewhat different framework has to be used,
depending for instance on basic valuations of the edges of the game graph.


\end{document}